\documentclass[accepted=2025-11-16, twocolumn, a4paper]{quantumarticle}
\pdfoutput=1
\usepackage[utf8]{inputenc}
\usepackage{amsmath, amssymb, amsthm}
\usepackage{graphicx}
\usepackage{ragged2e}
\usepackage{comment}
\usepackage{xcolor}
\usepackage[colorlinks = true,linkcolor = blue,urlcolor = blue,citecolor = blue,anchorcolor = blue]{hyperref}
\usepackage{url}
\usepackage{float}
\usepackage{multirow}
\usepackage{subcaption}
\usepackage[braket, qm]{qcircuit}

\newcommand{\braket}[2]{{\langle{#1}|{#2}\rangle}}
\newcommand{\ang}[2]{\angle^{#1}_{#2}}
\newcommand{\MCO}[3]{[#1]^{#2}_{#3}}
\newcommand{\sqO}[2]{[#1]_{#2}}

\newcommand{\abs}[1]{\lvert #1 \rvert}
\newcommand{\rpicl}{\Pi}
\newcommand{\rpi}[1]{\rpicl(#1)}
\newcommand{\rpisub}[1]{\rpicl_{#1}}
\newcommand{\rpisubu}[2]{\rpicl^{#2}_{#1}}

\newcommand{\rv}[2]{R_{#2}(#1)}

\newcommand{\rz}{R_{\hat{z}}}

\newcommand{\ry}{R_{\hat{y}}}

\newcommand{\rpigatesub}[1]{\gate{\mathrm{\rpisub{#1}}}}
\newcommand{\rpigatesubu}[2]{\gate{\mathrm{\rpisubu{#2}{#1}}}}
\newcommand{\rpigate}[1]{\gate{\rpi{#1}}}
\newcommand{\rvgate}[2]{\gate{\mathrm{\rv{#1}{#2}}}}

\newtheorem{lemma}{Lemma}

\newtheorem{theorem}{Theorem}

\renewcommand{\sec}[1]{\hyperref[sec:#1]{Section~\ref*{sec:#1}}}
\newcommand{\ssec}[1]{\hyperref[ssec:#1]{Section~\ref*{ssec:#1}}}
\newcommand{\fig}[1]{Figure~\hyperref[fig:#1]{\ref*{fig:#1}}}
\newcommand{\tab}[1]{Table~\hyperref[tab:#1]{\ref*{tab:#1}}}
\newcommand{\lem}[1]{Lemma~\hyperref[lem:#1]{\ref*{lem:#1}}}
\newcommand{\prop}[1]{\hyperref[prop:#1]{Proposition~\ref*{prop:#1}}}
\newcommand{\thm}[1]{Theorem~\hyperref[thm:#1]{\ref*{thm:#1}}}
\newcommand{\qc}[1]{Circuit~(\hyperref[qc:#1]{\ref*{qc:#1}})}
\newcommand{\apx}[1]{Appendix~\hyperref[apx:#1]{\ref*{apx:#1}}}
\newcounter{qcnum}
\newcommand{\qcref}[1]{\refstepcounter{qcnum}\tag{\theqcnum}\label{qc:#1}}

\begin{document}

\title{All You Need is pi: Quantum Computing with Hermitian Gates}

\author{Ben~Zindorf}
\email{ben.zindorf.19@ucl.ac.uk}
\affiliation{Department of Physics and Astronomy, University College London, Gower Street, WC1E 6BT London, United Kingdom}
\orcid{0000-0001-8630-3501}
\author{Sougato~Bose}
\affiliation{Department of Physics and Astronomy, University College London, Gower Street, WC1E 6BT London, United Kingdom}
\orcid{0000-0001-8726-0566}

\begin{abstract}
Universal gate sets for quantum computation, when single and two qubit operations are accessible, include both Hermitian and non-Hermitian gates. Here we {utilize the fact} that any single-qubit operator may be implemented as two Hermitian gates, and thus a purely Hermitian universal set is possible. This implementation can be used to prepare high fidelity single-qubit states in the presence of amplitude errors, and helps to achieve a high fidelity single-qubit gate decomposition using four Hermitian gates. An implementational convenience can be that non-identity single-qubit Hermitian gates are equivalent to $\pi$ rotations up to a global phase. We show that a gate set comprised of $\pi$ rotations about two fixed axes, along with the CNOT gate, is universal for quantum computation.
Moreover, we show that two $\pi$ rotations can transform the axis of any multi-controlled unitary, a special case being a single CNOT sufficing for any controlled $\pi$ rotation. These gates simplify the process of circuit compilation in view of their Hermitian nature. We exemplify by designing efficient circuits for a variety of controlled gates, and achieving a CNOT count reduction for the four-controlled Toffoli gate in LNN-restricted qubit connectivity.

\end{abstract}
\maketitle
\section{Introduction}
Pivotal requirements in the field of quantum computation are highly reliable gates, as well as {efficient} quantum circuits (i.e, {a low} number of fundamental gates {being used to implement} a given unitary operation). 
As it stands, {universal gate sets for quantum computation, consisting of fundamental one- and two-qubit gates, 
}
are composed of both Hermitian 
gates{, such as the CNOT, Hadamard, and Pauli gates}, as well as Non-Hermitian gates, such as the S and T gates 
\cite{nielsen_quantum_2002}. It is known that if you allow for three qubit gates, then Toffoli (which is Hermitian) allows for universal quantum computation along with the Hadamard. However, typically single and two qubit gates are used to implement the Toffoli. One can then ask a simple, but fundamental, question: can we do universal quantum computation purely by using {single and two qubit} Hermitian gates? As the CNOT is already Hermitian, the question effectively becomes as to whether all the required single-qubit unitaries can be Hermitian. As we will discuss in this paper, the answer is affirmative. 
The usage of Hermitian single-qubit gates then immediately has some positive implications for the implementations of quantum computation. 

 Firstly, {Hermitian single-qubit gates} can be implemented, up to a global phase, as $\pi$ rotations about any axis in the Bloch sphere, also called a $\pi$-pulse in experimental terminology \cite{bar-gill_solid-state_2013}. It effectively means that one switches on a Hamiltonian $g\hat{\sigma}\cdot\hat{v}$, where $g$ is a coupling strength, and $\hat{v}$ is the unit vector along any axis of the Bloch sphere, to act on the qubit for an amount of time $t=\pi/g$. For a reliable gate, the time $t$ has to be controlled very accurately. Typically, this time is a duration for which a voltage/microwave/laser pulse is applied to the qubit. Thus, one can set a control machine to generate one such exquisite pulse of duration giving a $\pi$ pulse to a very high precision, and then repeat the same pulse settings for {\em all} the gates, just by changing the axis $\hat{v}$ of the rotation in the Bloch sphere. It is plausible that in some technologies, the axes of rotation can be more precisely set than the time duration. If we have electromagnetic pulses of fixed duration, the axis of the Bloch sphere can be set to be anywhere in the $x-y$ plane by choosing their phase $\phi$ as
$\hat{v}=(\cos{\phi},-\sin{\phi})$ \cite{haroche_collge_2012}. Such rotations, along with Hadamard (which is also a $\pi$ rotation, albeit about a different axis), suffice. Moreover, as we will show, we require {\em only two {fixed}} axes for a universal Hermitian gate set, so that in any technology, one can focus on setting them precisely. In addition, composite pulses have been shown to significantly improve the fidelity of a $\pi$ pulse in the presence of pulse-strength errors (aka amplitude errors), while applying this to other angles is more challenging \cite{husain_further_2013,low_optimal_2014,tycko_fixed_1985,jones_nested_2013}. These robust $\pi$ pulses are a critical ingredient of dynamical decoupling
\cite{viola_dynamical_1999,uhrig_keeping_2007,souza_robust_2011,souza_robust_2012}, which is used to protect a qubit from the environment and reduce its decoherence rate. For the same reason, $\pi$ pulses have been extensively refined.
For the task of suppressing amplitude errors, we provide a sequence of two $\pi$ rotations which improves the fidelity of single-qubit state preparation. Similarly, we provide a high fidelity decomposition of arbitrary single-qubit gates using four $\pi$ rotations in the presence of amplitude errors.


{Secondly, a} positive implication stems from the inherent characteristic of Hermitian gates being their own inverse. This allows them to be used as a symbolic tool to simplify circuit optimization and compilation processes in many cases. As we will show, the Hermitian gates provide a useful framework, especially for the task of decomposing controlled gates. For example, it allows us to change the axis of rotation of any Multi-Controlled unitary to any other axis. This allows one to change the CNOT to any controlled $\pi$ rotation. We utilize this framework to provide a decomposition of any controlled $U(4)$ gate using a small number of CNOT gates. The Hermitian gates can then be replaced with gates from any choice of universal set in an efficient manner.
{In addition, we use the same framework to find efficient decompositions of various types of multi-controlled single-target gates with a small number of controls. While some of our circuits require a CNOT count which matches the state-of-the-art, others provide generalized versions of these gates, with no increase of CNOT count. In the case of a Toffoli gate with four controls under the restrictions of LNN/chain connectivity, our methods resulted in a reduction of CNOT gates compared to the state-of-the-art \cite{nemkov_efficient_2023}.}

{Here we summarize the results of the paper. 
We start by providing a geometrically motivated description of Hermitian $\pi$-rotations and using this to demonstrate the implications of known mathematical concepts, such as Householder reflections
\cite{householder_unitary_1958}, for single qubit rotations -- demonstrating the ability to use two Hermitian gates to decompose any single-qubit operator, or to transform the rotation axis of a given one. 
We then present a variety of universal gate sets comprised of {\em restricted} subsets of all Hermitian gates, smallest of which requires only the ability to apply the CNOT gate, and to rotate the Bloch sphere about {\em two fixed axes} by the {\em fixed angle} $\pi$ for universal quantum computation. 
We then present our $\pi$-rotation sequences which provide a quadratic suppression of amplitude errors, using four $\pi$-rotations for any single-qubit unitary, and only two  $\pi$-rotations for state preparation.
We then show that the single-qubit axis transformations can be applied to the target operator of controlled gates, which immediately provides improved decompositions of multi-controlled gates compared to state of the art. This also implies
that one CNOT gate suffices to implement any singly-controlled $\pi$-rotation, while any other choice of angle requires two CNOT gates, providing an alternative geometrical description of Lemma 5.5 in \cite{barenco_elementary_1995}.
We continue to show that Hermitian gates can be used as an intuitive symbolic tool for synthesizing multi-controlled gates, specifically resulting in a reduced CNOT count of the four-controlled Toffoli gate, compared to the state-of-the-art decompositions which relied on machine learning or exhaustive search. Finally we demonstrate the ability to add a control to a given unitary using the same formalism, exemplifying on $U(4)$ gates.}
\\

\section{Single qubit operators}\label{ssec:single_qubit_op}
We define a $\pi$-rotation as the Hermitian operator $\rpi{\hat{v}}:=i\rv{\pi}{\hat{v}}$, such that $\rv{\pi}{\hat{v}}\in SU(2)$ applies a rotation by $\pi$ about the axis $\hat{v}=(v_{x},v_{y},v_{z})$. The axis can be represented by the standard spherical coordinates 
$\hat{v}(\theta,\phi)=(\sin{\theta}\cos{\phi},\sin{\theta}\sin{\phi},\cos{\theta})$. 
A $\pi$-rotation about the axis $\hat{v}$ can thus be written in the two following ways:
\[
\rpi{\hat{v}}
=
\left(
\begin{matrix}
v_z & v_x-iv_y\\
v_x+iv_y & -v_z  
\end{matrix}
\right)
\]
\[
\rpi{\theta,\phi}:=\rpi{\hat{v}(\theta,\phi)}=
\left(
\begin{matrix}
\cos{\theta} & e^{-i\phi}\sin{\theta}\\
e^{i\phi}\sin{\theta} & -\cos{\theta} 
\end{matrix}
\right)
\]
 Applying two $\pi$-rotations about the {\em same axis} results in a $2\pi$-rotation, and thus $(\rpi{\hat{v}})^2=-\rv{2\pi}{\hat{v}}=I$.\\ 
 { Interestingly,} any $SU(2)$ operator can be implemented as two $\pi$-rotations{, as we state in \lem{2_rpi} and prove in \apx{pi_lemmas}}. {Since $SU(2)$ is a double cover of $SO(3)$, this Lemma can also be inferred from \cite{brezov_vector_2012, donchev_compositions_2015} in the mathematics literature, however, its implications for quantum computing seem to have been overlooked.} 
We mark $\ang{\hat{v}_1}{\hat{v}_2}$ as the angle between two unit vectors $\hat{v}_1,\hat{v}_2$, and $\hat{R}_{\hat{v}}(\theta)$ as the three dimensional rotation matrix by angle $\theta$ about the axis $\hat{v}$.
\begin{lemma}\label{lem:2_rpi}
Any $\rv{\lambda}{\hat{v}}\in SU(2)$ operator can be implemented as $\rpi{\hat{v}_2}\rpi{\hat{v}_1}$ 
with $\hat{v}_1$ as any unit vector perpendicular to $\hat{v}$, and $\hat{v}_2 = \hat{R}_{\hat{v}}(\frac{\lambda}{2})\hat{v}_1$ with $\frac{\lambda}{2}\in (-\pi,\pi]$.
\end{lemma}

As $U(2)$ operators can be implemented as $SU(2)$ up to a global phase, \lem{2_rpi} shows that any single qubit operator can be implemented as two $\pi$-rotations (\fig{sphere}).
This implementation requires four parameters - two for each $\rpicl$ gate, while $SU(2)$ gates require only three. However, one of the vectors $\hat{v}_1,\hat{v}_2$ can always be chosen on the $xy$ plane, so that it is perpendicular to both the $\hat{z}$ axis and the vector $\hat{v}$. Therefore, one of the $\pi$-rotations only requires one parameter, since it can be described as $\rpi{\frac{\pi}{2},\phi}$. Moreover, since $\rpi{-\hat{v}} = -\rpi{\hat{v}}$ for any $\hat{v}$, the required $\pi$-rotations can always be chosen from one half of the Bloch sphere, e.g. the angles defining $\rpi{\theta,\phi}$ can be constrained to $\theta,\phi \in [0,\pi)$, {potentially introducing a $-1$ global phase to the $SU(2)$ gate.}\\
Arbitrary rotations about the $\hat{x},\hat{y}$, and $\hat{z}$ axes require a single parameter. In the $\pi$-rotation decomposition, this can be achieved by choosing one of the vectors $\hat{v_1},\hat{v_2}$ to be in a fixed position. For example, a rotation about $\hat{z}$ requires two vectors on the $xy$ plane, and thus a comfortable choice would be the $\hat{x}$ axis since $\rpi{\hat{x}}=X$, and $\rv{\lambda}{\hat{z}}$ can be implemented as $\rpi{\frac{\pi}{2},\frac{\lambda}{2}}X$, or as $X\rpi{\frac{\pi}{2},-\frac{\lambda}{2}}$. For this case, the only required $\pi$-rotations are $\rpi{\frac{\pi}{2},\phi}$, with $\phi\in [0,\frac{\pi}{2}]$.
\begin{figure}[h]
    \centering
    \begin{subfigure}{.49\linewidth}
    \centering
    \includegraphics[width=0.99\linewidth]{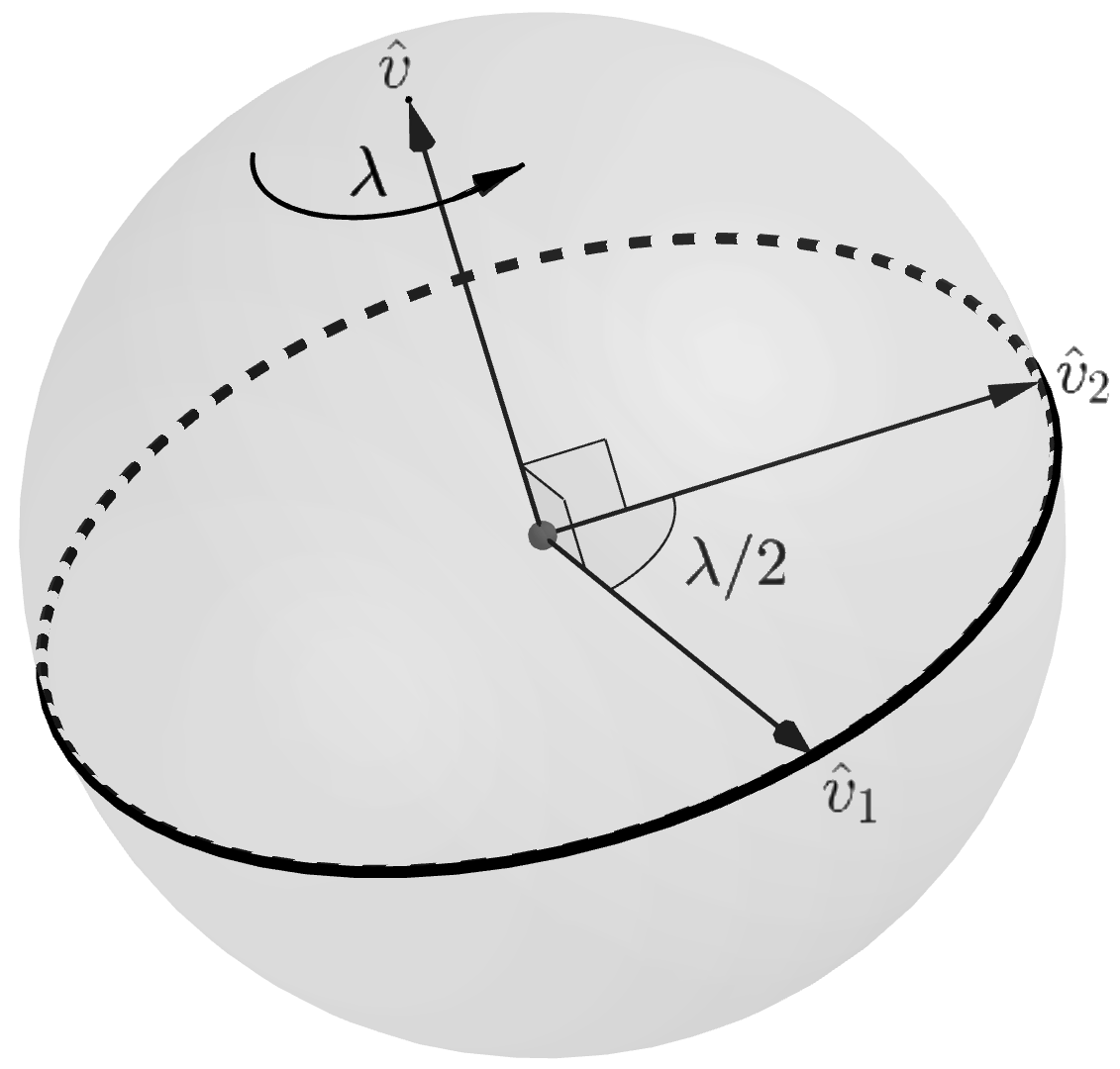}
    \caption{}
    \label{fig:sphere}
    \end{subfigure}
    \begin{subfigure}{.49\linewidth}
    \centering
    \includegraphics[width=0.99\linewidth]{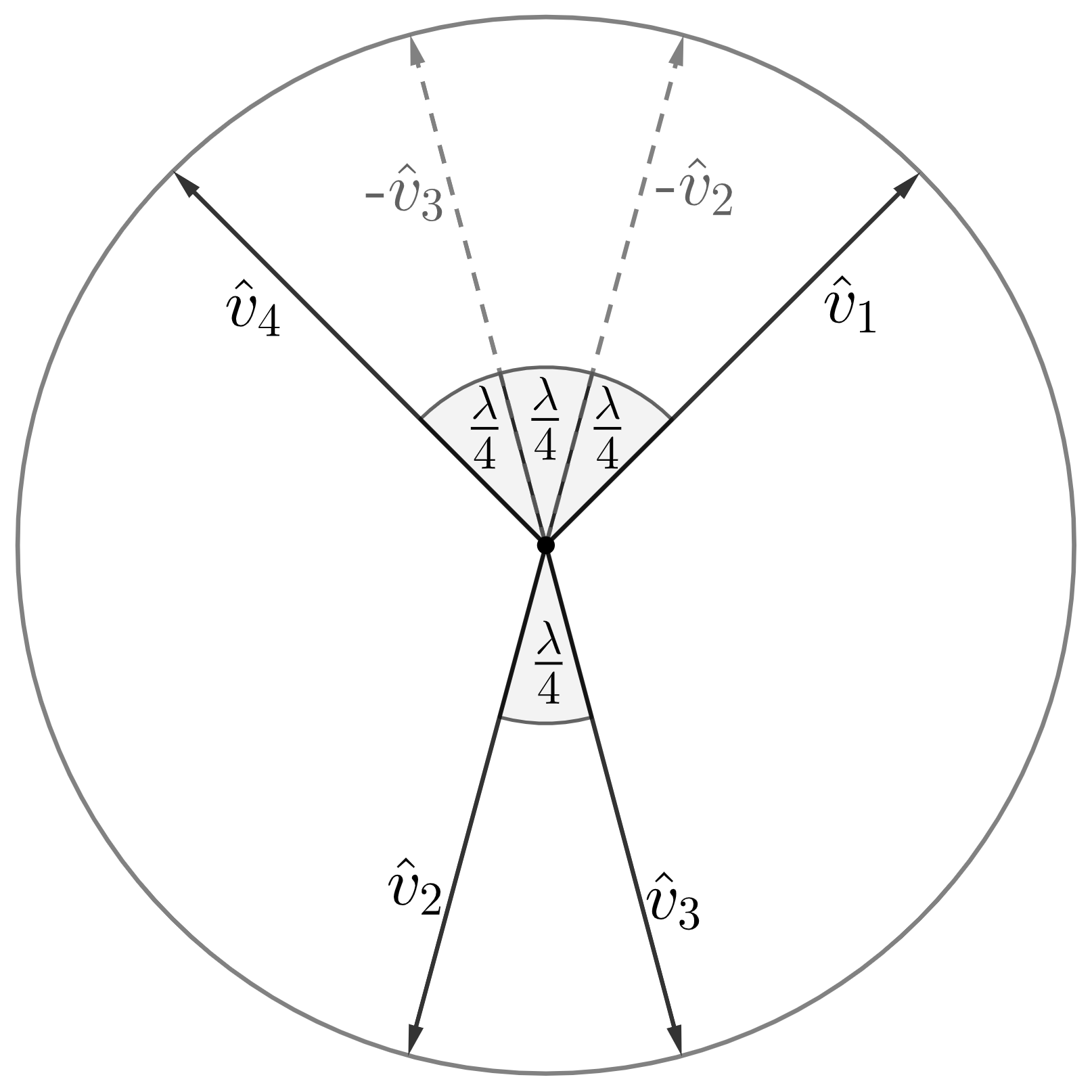}
    \caption{}
    \label{fig:vecs}    
    \end{subfigure}
    \caption{(a) A rotation by $\lambda$ about the vector $\hat{v}$ can be achieved by a $\pi$-rotation about $\hat{v}_1$ followed by a $\pi$-rotation about $\hat{v}_2$. {(b) Describing the chosen axes (on the plane perpendicular to $\hat{v}$) for our four $\Pi$ gates decomposition.}}
    \label{fig:sphere_vecs}    
\end{figure}

{

 {Single qubit $\pi$-rotations can be used to transform the rotation axis of any given single-qubit gate to any other axis, as we show next.
 We define the set of unit vectors $M(\hat{v}_1,\hat{v}_2)$ for any $\hat{v}_1,\hat{v}_2$, such that $\hat{v}_M \in M(\hat{v}_1,\hat{v}_2)$ iff $\hat{v}_2 = \hat{R}_{\hat{v}_M}(\pi)\hat{v}_1$. In case $\hat{v}_1=-\hat{v}_2$, the set $M(\hat{v}_1,\hat{v}_2)$ holds all of the vectors perpendicular to $\hat{v}_1$, and otherwise, the set holds the vector $\hat{v}_M=\frac{\hat{v}_1+\hat{v}_2}{\lvert \hat{v}_1+\hat{v}_2 \rvert}$ which is in the middle between $\hat{v}_1$ and $\hat{v}_2$. Therefore,  $M(\hat{v}_1,\hat{v}_2)\not=\emptyset$ for any choice of $\hat{v}_1,\hat{v}_2$. 
\begin{lemma}\label{lem:sq_trans} 
$
\rpi{\hat{v}}=
\rpi{\hat{v}_{M}}
\rpi{\hat{v}'}
\rpi{\hat{v}_{M}}
$
for any chosen unit vectors $\hat{v},\hat{v}'$, and $\hat{v}_M \in M(\hat{v},\hat{v}')$.
\end{lemma}
\begin{proof}
{Immediate from \lem{2_rpi}, noting that $\rpi{\hat{v}'}
\rpi{\hat{v}_{M}}$ and $\rpi{\hat{v}_{M}}
\rpi{\hat{v}}$ implement the same rotation.}
\end{proof}  
The following lemma{, which we prove in \apx{pi_lemmas},} allows to transform the axis of a given rotation to any other axis, using two $\pi$-rotations.
\begin{lemma}\label{lem:1q_transform} 
$
\rv{\lambda}{\hat{v}}
=
\rpi{\hat{v}_{M}}
\rv{\lambda}{\hat{v}'}
\rpi{\hat{v}_{M}}
$
for any angle $\lambda$, unit vectors $\hat{v},\hat{v}'$, and $\hat{v}_M \in M(\hat{v},\hat{v}')$.
\end{lemma}

}
}
{
\section{Hermitian Sets}
The Hermitian CNOT gate, along with arbitrary single-qubit operators are universal for quantum computation. Since any single qubit operator can be implemented using two $\pi$-rotations as per \lem{2_rpi}, the Hermitian set $\{CNOT,\rpi{\theta,\phi}\}$ with $\theta,\phi \in [0,\pi)$ is universal as well. Similarly, the universality of $\{CNOT,H,\rz\}$ provides the Hermitian universal set $\{CNOT,H,\rpi{\frac{\pi}{2},\phi}\}$ with $\phi\in [0,\frac{\pi}{2}]$. 
The latter Hermitian set requires significantly less available axes for $\pi$-rotations than the former.
Since $S:=P(\frac{\pi}{2})=\rpisub{S}X$, and $T:=P(\frac{\pi}{4})=\rpisub{T}X$ up to a global phase, with 
$\rpisub{S}:=\rpi{\frac{\pi}{2},\frac{\pi}{4}}$, $\rpisub{T}:=\rpi{\frac{\pi}{2},\frac{\pi}{8}}$ {and $P(\psi):= 
e^{i\frac{\psi}{2}}\rv{\psi}{\hat{z}}$ as the phase gate,} the Hermitian sets $\{CNOT,H,X,\rpisub{S}\}$ and $\{CNOT,H,X,\rpisub{T}\}$ span the Clifford and and the universal Clifford+T \cite{bravyi_universal_2005} groups.
The universality of the set $\{CNOT,H,X,\rpisub{T}\}$ is particularly interesting since it shows that any quantum operator can be realised using CNOT gates along with $\pi$-rotations about the three fixed axes $\hat{x}=\hat{v}(\frac{\pi}{2},0)$, $\hat{h}=\hat{v}(\frac{\pi}{4},0)$ and $\hat{v}_T=\hat{v}(\frac{\pi}{2},\frac{\pi}{8})$ {(for a visual representation see \cite{zindorf_multi-controlled_2025})}.
We now show that the CNOT gate along with $\pi$-rotations around the {\em two fixed axes} $\hat{h}$ and $\hat{v}_T$ is universal for quantum computation as well. 
We simply remove the $X$ gate from the Hermitian version of Clifford+T {(can be applied for the Hermitian version of Clifford as well)}, since it can be implemented using other gates from this set as follows.
\[
\resizebox{0.9\linewidth}{!}{
 \Qcircuit @C=0.5em @R=0.em @!R { 
	 	\nghost{{q} :  } & \lstick{{q} :  } & \targ & \qw & \targ & \qw & \qw \\
	 	\nghost{{a} :  } & \lstick{{a} :  } & \ctrl{-1} & \rpigatesub{T} & \ctrl{-1} & \rpigatesub{T} & \qw \\
 }
\hspace{5mm}\raisebox{-3mm}{=}\hspace{0mm}
\Qcircuit @C=0.5em @R=0.2em @!R { 
	 	\nghost{{q} :  } & \lstick{{q} :  } & \targ & \gate{\mathrm{H}} & \ctrl{1} & \gate{\mathrm{H}} & \targ & \gate{\mathrm{H}} & \ctrl{1} & \gate{\mathrm{H}} & \qw \\
	 	\nghost{{a} :  } & \lstick{{a} :  } & \ctrl{-1} & \qw & \targ & \qw & \ctrl{-1} & \qw & \targ & \qw & \qw \\
 }
 }
\]

With qubit $a$ used as an ancilla in an arbitrary state, to implement {the $X$ gate on $q$}. If the ancilla is in the state $\ket{0}$, the first CNOT gate can be removed from each these circuits, and if the ancilla is in state $\ket{1}$, the $X$ gate can be implemented using one CNOT {controlled by $a$}. {In \apx{toffolispi}, we provide implementations of the Toffoli gate using either the $\{CNOT,H,X,\rpisub{T}\}$ set or the $\{CNOT,H,\rpisub{T}\}$ set.}
We have presented various universal sets of gates since existing circuits may be set up to use any local unitary or specific sets of unitaries. These could then be easily transformed to be implemented in terms of their Hermitian counterparts as given in \tab{herm_sets}. Note that more {\em restrictive} a set of gates is, one may require a larger number of gates for the given operator as should be expected.
\begin{table}[H]
    \centering
    \resizebox{0.99\linewidth}{!}{
    \begin{tabular}{ 
    |c||c||c|c|}
\hline
 Group&Standard Set&Hermitian Set&Range\\
 \hline
 Universal    & $\{CNOT,U\}$ & $\{CNOT,\rpi{\theta,\phi}\}$   & $\theta,\phi \in [0,\pi)$ \\
& $\{CNOT,H,\rz\}$     & $\{CNOT,H,\rpi{\frac{\pi}{2},\phi}\}$   & $\phi\in [0,\frac{\pi}{2}]$ \\
 & $\{CNOT,H,T\}$   & $\{CNOT,H,\rpisub{T},X\}$  &  \\
  & $\{CNOT,H,T\}$   & $\{CNOT,H,\rpisub{T}\}$  &  \\
 \hline
 Clifford     & $\{CNOT,H,S\}$   & $\{CNOT,H,\rpisub{S},X\}$   &  \\
      & $\{CNOT,H,S\}$   & $\{CNOT,H,\rpisub{S}\}$   &  \\
 \hline
    \end{tabular}
    }
    \caption{Commonly used gate sets and their Hermitian counterparts.}
    \label{tab:herm_sets}
\end{table}
    }

{
\section{Amplitude error suppression}\label{ssec:single_qubit_es}
We now consider the benefits of using $\pi$ rotations in the presence of systematic amplitude errors. Similarly to \cite{husain_further_2013,low_optimal_2014,tycko_fixed_1985,jones_nested_2013}, we assume that while intending to apply an ideal rotation $\rv{\lambda}{\hat{v}}$, these errors change the rotation to $\rv{(1+\epsilon)\lambda}{\hat{v}} = \rv{\lambda}{\hat{v}}\rv{\epsilon\lambda}{\hat{v}}$ with a small unknown $\epsilon$. We wish to minimize the infidelity which is defined as $\mathcal{I} = 1-\abs{tr(VU^\dagger)}/2$, where $U$ and $V$ are the ideal and non-ideal rotations, respectively.
{For this task, we use the freedom given by the decomposition in \lem{2_rpi} to arbitrarily choose one of the two required vectors from the plain perpendicular to $\hat{v}$ to implement two instances of $\rv{\tfrac{\lambda}{2}}{\hat{v}}$ in different ways. An additional benefit unique to Hermitian $\pi$-rotation is the freedom to replace a $\Pi(\hat{v}_j)$ gate with $\Pi(-\hat{v}_j)$ while only introducing a global phase of $-1$. Utilizing these freedoms for the task of minimizing the infidelity has resulted in the following.}

We use \lem{2_rpi} to decompose $\rv{\lambda}{\hat{v}}$ using four $\pi$-rotations as $\rpisub{4}\rpisub{3}\rpisub{2}\rpisub{1}$ such that $\rpisub{j}=\rpi{\hat{v}_j}$ and $\rpisub{4}\rpisub{3}=\rpisub{2}\rpisub{1}=-\rv{\tfrac{\lambda}{2}}{\hat{v}}$. The axis $\hat{v}_1$ is perpendicular to $\hat{v}$, and the rest are defined as $\hat{v}_{j+1} = \hat{R}_{\hat{v}}(\tfrac{\lambda}{4}+j\pi)\hat{v}_{j}$ (\fig{vecs}).
In the presence of amplitude errors, each $\rpisub{j}$ gate is replaced by $\rpisub{j}\rv{\epsilon\pi}{\hat{v}_j}$. While we had the freedom to choose $\hat{v}_3$ as any vector perpendicular to $\hat{v}$, our specific choice is beneficial, as it provides the following commutation rules which can be obtained using \lem{1q_transform}: $\rv{\epsilon\pi}{\hat{v}_4}
\rpisub{3}= \rpisub{3}\rv{\epsilon\pi}{-\hat{v}_2}$ and $\rpisub{2}
\rv{\epsilon\pi}{\hat{v}_1} = 
\rv{\epsilon\pi}{-\hat{v}_3}\rpisub{2}$. The four imperfect $\Pi$ gates  decomposition can therefore be written as $\rpisub{4}
\rpisub{3}
E
\rpisub{2}
\rpisub{1}=\rv{\tfrac{\lambda}{2}}{\hat{v}}E\rv{\tfrac{\lambda}{2}}{\hat{v}}$ with
\[
E := \rv{\epsilon\pi}{-\hat{v}_2}
\rv{\epsilon\pi}{\hat{v}_3}
\rv{\epsilon\pi}{\hat{v}_2}
\rv{\epsilon\pi}{-\hat{v}_3},
\]
{
and the infidelity is given by 
$\mathcal{I} = 1-\abs{tr(E)}/2$.
From the definition of $\hat{v}_2$ and $\hat{v}_3$, and since the trace is invariant under a change of basis, we can solve for $\hat{v}=\hat{z}$, $\hat{v}_2=\hat{y}=v(\tfrac{\pi}{2},\tfrac{\pi}{2})$ and $\hat{v}_3=v(\tfrac{\pi}{2},\tfrac{\pi}{2}+\tfrac{\lambda}{4})$.
Using \lem{2_rpi} we can write for this choice
$\rv{\epsilon\pi}{\hat{v}_2} = 
\Pi(\frac{\pi\epsilon}{2},0)Z$
and
$\rv{\epsilon\pi}{\hat{v}_3} = 
\Pi(\frac{\pi\epsilon}{2},\frac{\lambda}{4})Z$.
Setting $a=\cos^2(\tfrac{\pi\epsilon}{2})+e^{i\frac{\lambda}{4}}\sin^2(\tfrac{\pi\epsilon}{2})$ and $b=(e^{-i\frac{\lambda}{4}}-1)\cos(\tfrac{\pi\epsilon}{2})\sin(\tfrac{\pi\epsilon}{2})$ provides the following from the definition of $\Pi(\theta,\phi)$
{
\begin{align*}
&tr(E) = 
tr((Z\Pi(\tfrac{\pi\epsilon}{2},0)
\Pi(\tfrac{\pi\epsilon}{2},\tfrac{\lambda}{4}))^2)
=\\
&tr\left(
\left(
\begin{matrix}
a &
b\\
b^* &
-a^* 
\end{matrix}
\right)^2
\right)=
a^2+{a^*}^2+2b^*b
=\\
&2\cos^2(\tfrac{\pi\epsilon}{2})\left(\cos^2(\tfrac{\pi\epsilon}{2})+2\sin^2(\tfrac{\pi\epsilon}{2})\right)
+2\cos(2\tfrac{\lambda}{4})\sin^4(\tfrac{\pi\epsilon}{2})
.
\end{align*}
}
{
Substituting $\cos^2(\tfrac{\pi\epsilon}{2}) = 1-\sin^2(\tfrac{\pi\epsilon}{2})$
and
$\cos(2\tfrac{\lambda}{4}) = 1-2\sin^2(\tfrac{\lambda}{4})$
provides
$tr(E)=
2-4\sin^4(\tfrac{\pi\epsilon}{2})\sin^2(\tfrac{\lambda}{4})$. We get $\abs{tr(E)}=tr(E)$ when the amplitude error is sufficiently small ($\abs{\epsilon}<0.6$), 
and therefore, by definition, $\mathcal{I} = 2\sin^4(\tfrac{\epsilon\pi}{2})\sin^2(\tfrac{\lambda}{4})\approx \epsilon^4\tfrac{\pi^4}{8}\sin^2(\tfrac{\lambda}{4})$.}}
This provides a quadratic improvement compared to the infidelity of the direct implementation $R_{\hat{v}}((1+\epsilon)\lambda)$ given by $1-\cos(\tfrac{\epsilon\lambda}{2})\approx \epsilon^2\lambda^2/8$. We provide a numerical analysis in \fig{gate_fid_graph}, comparing these implementations with the Euler ZYZ, and the two $\Pi$ gate (from \lem{2_rpi} with $\hat{v}_2\Rightarrow -\hat{v}_2$) decompositions. As can be seen, the latter two decompositions result in a similar infidelity for the chosen rotations.

\begin{figure}[h]
    \centering
    \includegraphics[width=1\linewidth]{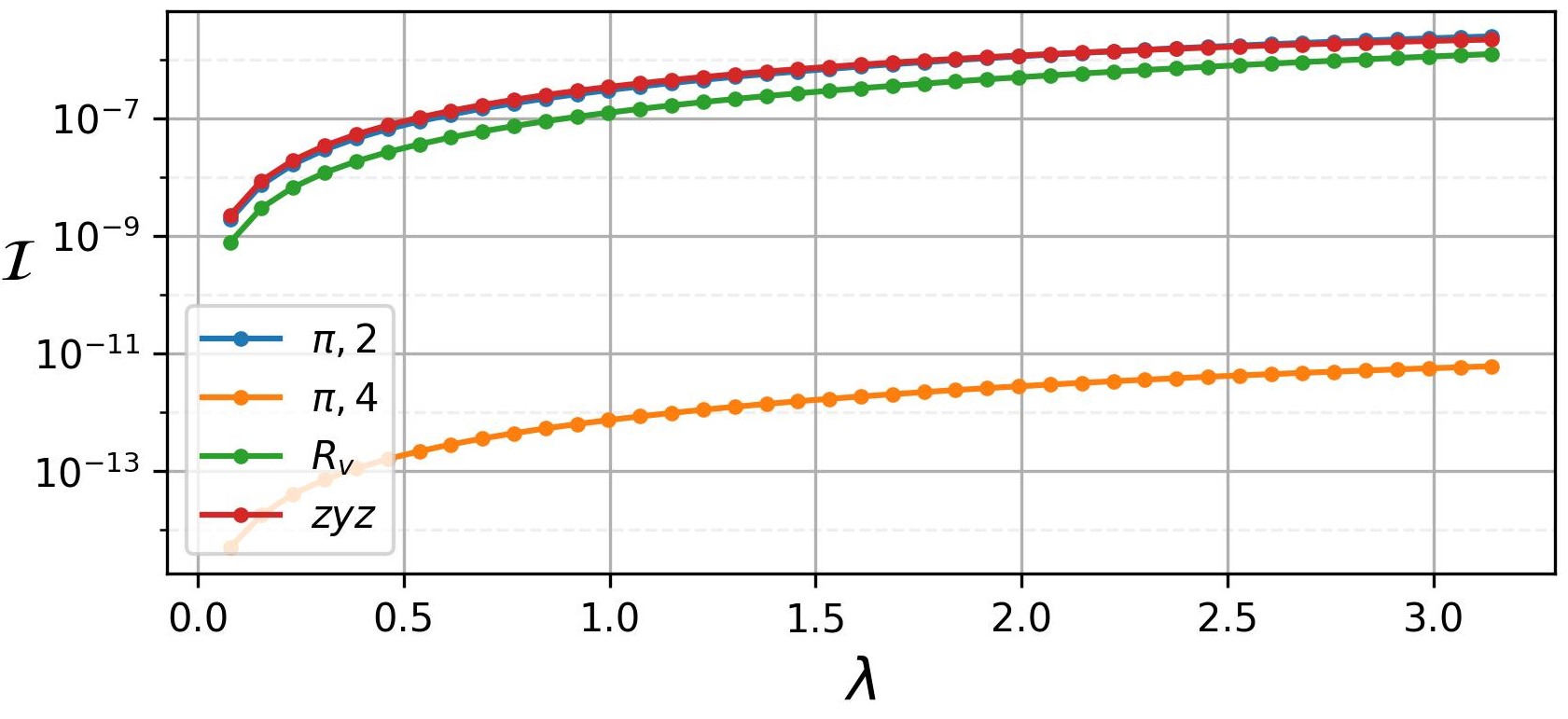}
    \caption{{Numerical gate infidelity (log scale) as a function of the rotation angle $\lambda$ with a fixed axis $\hat{v}(\tfrac{\pi}{2},\tfrac{\pi}{16})$ and an amplitude error $\epsilon=10^{-3}$. Comparing the two $\Pi$, four $\Pi$, Euler ZYZ and the direct $R_v$ implementations.}}
    \label{fig:gate_fid_graph}
\end{figure}
Although other four-step decompositions have been proposed \cite{wimperis_iterative_1991,wimperis_broadband_1994,gevorgyan_ultrahigh-fidelity_2021} in the context of composite pulses, achieving similar benefits, we are not aware of such a decomposition which achieves {\em arbitrary} rotations using four $\pi$ rotations/pulses.

By relaxing the requirement to implement a specific rotation, we now focus on transforming an arbitrary initial single-qubit state $\ket{\Psi_A}$ to an arbitrary final state $\ket{\Psi_B}$ with the corresponding state vectors $\hat{v}_A,\hat{v}_B$ respectively. It has previously been shown that a single $\pi$ rotation about the axis $\tfrac{\hat{v}_A+\hat{v}_B}{\abs{\hat{v}_A+\hat{v}_B}}$ suffices for this task \cite{shim_single-qubit_2013}. Moreover, since $\pi$ is the only angle which allows to transform $\ket{0}$ to $\ket{1}$ in one step, it immediately follows that $\pi$ is the {\em only} rotation angle which allows to transform {\em any} initial state to {\em any} final state. 
Generally, many trajectories can be chosen between $\hat{v}_A$ and $\hat{v}_B$, for example, a geodesic trajectory can be achieved by a $R_{\hat{v}_\perp}(\lambda)$ rotation with $\hat{v}_\perp$ perpendicular to both $\hat{v}_A,\hat{v}_B$ and $\lambda=\ang{\hat{v}_A}{\hat{v}_B}$.
In the presence of amplitude errors, we find that using the two $\Pi$ gate decomposition from \lem{2_rpi} can provide a quadratic improvement to the fidelity of the final state, compared to the geodesic approach.
The state infidelity can be written as $\mathcal{I}_{\ket{}}=1-\abs{\braket{\Psi_B}{\Psi'_B}}^2$, with $\ket{\Psi'_B}$ as the final state achieved by imperfect rotations.
We decompose $R_{\hat{v}_\perp}(\lambda)$ as $\Pi_2\Pi_1$ using \lem{2_rpi} with $\hat{v}_2\Rightarrow -\hat{v}_2$ as before, while taking advantage of the fact that $\hat{v}_1$ can be chosen as any vector perpendicular to $\hat{v}_\perp$. We find that the ideal choice in this case is $\hat{v}_1=\hat{R}_{\hat{v}_\perp}(\tfrac{\pi}{2}+\tfrac{\lambda}{4})\hat{v}_A$. As can be seen in \fig{state_fid}, this provides a quadratic improvement compared to the geodesic trajectory, the Euler ZYZ, and the single $\Pi$ decompositions. 
}
\begin{figure}[h]
    \centering
    \includegraphics[width=1\linewidth]{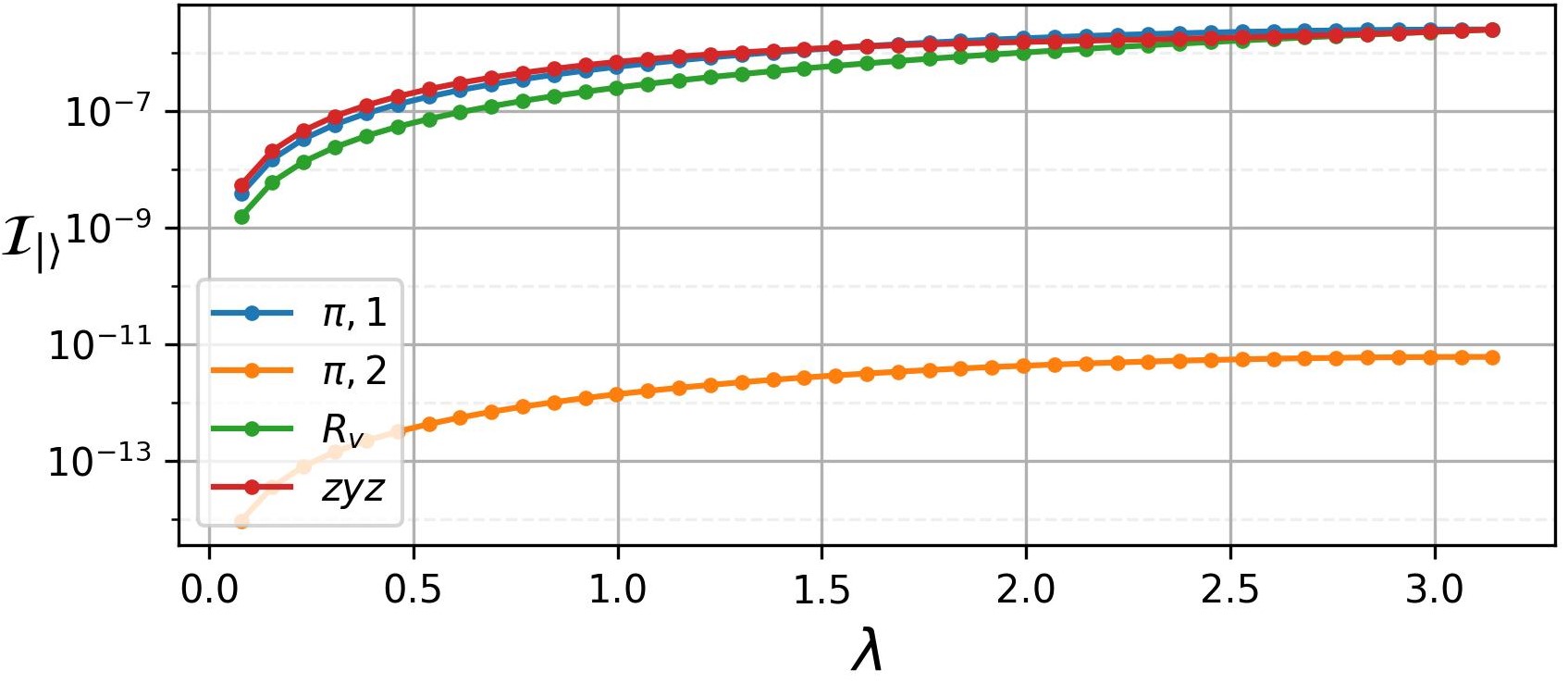}
    \caption{{Numerical state infidelity (log scale) as a function of $\lambda$, transforming the $\ket{0}$ state to a final state defined by the state vector $\hat{v}_B=\hat{v}(\lambda,\tfrac{\pi}{2})$. Comparing the single $\Pi$, two $\Pi$, Euler ZYZ and the geodesic ($R_{v_\perp}$) implementations with  amplitude error $\epsilon=10^{-3}$.}}
    \label{fig:state_fid}
\end{figure}
 \section{Controlled $\rpicl$ Gates}
In this section we show that a single-controlled $U\in U(2)$ gate can be implemented using one $CNOT$ gate \textit{iff}  $U=e^{i\psi}\rpi{\hat{v}}$ for any unit vector $\hat{v}$ and any angle $\psi$. Furthermore, we
provide a method to transform the rotation axis of any multi-controlled $U(2)$ gate to any other axis using two single qubit $\pi$-rotations.

We use $\MCO{\cdot}{C}{t}$ to mark a multi-controlled gate with a control set $C$ {holding $n$ qubits}, and a target {qubit} $t$. Similarly, $\sqO{\cdot}{t}$ marks a single qubit gate applied on $t$.
\begin{lemma}\label{lem:MC_transform} 
$\MCO{\omega\rv{\lambda}{\hat{v}}}{C}{t}=\sqO{\rpi{\hat{v}_{M}}}{t}\MCO{\omega\rv{\lambda}{\hat{v}'}}{C}{t}\sqO{\rpi{\hat{v}_{M}}}{t}$ with $\omega := e^{i\psi}$
for any angles $\lambda,\psi$, unit vectors $\hat{v},\hat{v}'$, and $\hat{v}_M \in M(\hat{v},\hat{v}')$.
\end{lemma}
\begin{proof}
{
The gate $\MCO{\omega\rv{\lambda}{\hat{v}'}}{C}{t}$ applies $\omega\rv{\lambda}{\hat{v}'}$ on qubit $t$ if set $C$ is in state $\ket{11..1}$ ("on"), and $I$ otherwise.
{For $\sqO{\rpi{\hat{v}_{M}}}{t}
\MCO{\omega\rv{\lambda}{\hat{v}'}}{C}{t}
\sqO{\rpi{\hat{v}_{M}}}{t}$, }
if set $C$ is "on", we get $\omega\rpi{\hat{v}_{M}}\rv{\lambda}{\hat{v}'}\rpi{\hat{v}_{M}}=\omega\rv{\lambda}{\hat{v}}$ from \lem{1q_transform}, and $\rpi{\hat{v}_{M}}\rpi{\hat{v}_{M}}=I$ otherwise. 
}
\end{proof}
{Existing work which addresses low CNOT count implementations of multi-controlled gates largely distinguish between the most general form of these gates, allowing for any rotation angle $\lambda$ about any axis $\hat{v}$, and the less general form, allowing any angle $\lambda$ about a fixed axis $\hat{v}'$. 
Focusing on the latter allows to simplify the problem, which may result in a lower gate count compared to the former. However, \lem{MC_transform} suggests that this lower gate count is applicable to both, at the cost of two additional single-qubit $\Pi$ gates. 
These $\Pi$ gates are applied on both sides of a given implementation of the multi-controlled gate, without requiring any further adjustments.
For example, the authors of \cite{vale_circuit_2024} show that any multi-controlled $SU(2)$ gate can be implemented using $20n+O(1)$ CNOT gates, while only $16n+O(1)$ CNOT gates are required if the rotation is about the $\hat{v}'=z$ axis. 
\lem{MC_transform} immediately implies that any multi-controlled $SU(2)$ gate can be implemented using $16n+O(1)$ CNOT gates -- a $20\%$ improvement over the state-of-the-art.
We note that the axis transformation from \lem{MC_transform} is not restricted to multi-controlled $SU(2)$ gates, but depending on the choice of $\psi$, it can be applied to multi-controlled $U(2)$ as well. This also allows to use well studied implementations of the multi-controlled Pauli-X gate to implement a multi-controlled Hadamard with the same CNOT count, by simply noting that both X and H are Hermitian $\pi$-rotations, albeit about a different axis. {For example, this would allow the results of \cite{rosa_optimizing_2025} to be modified, using our \lem{MC_transform}, resulting in a multi-controlled Hadamard cost of $12n$ instead of $32n$ -- an immediate $62.5\%$ reduction (which was already known from \cite{zindorf_efficient_2025} with a more relaxed ancilla requirement).
}

For the case of a single control, }  
by setting $\lambda=\pi, \hat{v}'=\hat{x}, C=\{c\}$, and replacing  $\psi$ with $\psi+\frac{\pi}{2}$, \lem{MC_transform} shows that {\em any} $\MCO{e^{i\psi}\rpi{\hat{v}}}{c}{t}$ gate {(e.g. controlled Hadamard)} can be implemented using a single CNOT gate simply by choosing $\hat{v}_M \in M(\hat{v},\hat{x})$ as:
\[
 \resizebox{0.9\linewidth}{!}{
\Qcircuit @C=1.0em @R=0.2em @!R { 
	 	\nghost{{c} :  } & \lstick{{c} :  }   & \ctrl{1} & \qw \\
	 	\nghost{{t} :  } & \lstick{{t} :  }  & \gate{e^{i\psi}\rpi{\hat{v}}} & \qw \\
 }
  \hspace{5mm}\raisebox{-4.5mm}{=}\hspace{0mm}
 \Qcircuit @C=1.0em @R=0.3em @!R { 
	 	\nghost{{c} :  } & \lstick{{c} :  }   & \gate{P(\psi)} & \ctrl{1} & \qw & \qw \\
	 	\nghost{{t} :  } & \lstick{{t} :  }  & \rpigate{\hat{v}_M} & \targ & \rpigate{\hat{v}_M} & \qw \\
 }
}
 \qcref{crpi}
 \]

{In \lem{rpi_if_1cx} (\apx{pi_lemmas}), we} show that these controlled $\pi$-rotations are the {\em only} controlled gates which can be implemented using a single CNOT gate. 
Combining 
{this}, and \lem{MC_transform} for a single control, proves the following theorem. {This theorem can also be realized from Lemma 5.5 in \cite{barenco_elementary_1995}.}
\begin{theorem} \label{thm:1CNOT_iff}
A single-controlled unitary gate $\MCO{U}{c}{t}, U\in U(2)$ can be implemented using one $CNOT$ gate iff $U=e^{i\psi}\rpi{\hat{v}}$ for any unit vector $\hat{v}$ and any angle $\psi$.
\end{theorem}

This realization provides sufficient motivation to use controlled $\pi$-rotations as elementary gates when constructing large circuits, especially if the goal is to minimize the count of CNOT gates. {We demonstrate this in the next section.}
{
\\

\section{Controlled SU(2)/U(2) gates}\label{sec:Controlled_U2}
In this section we utilize the Hermitian gate framework in order to implement controlled gates using a small CNOT count. We emphasize the simplicity in which our structures are developed, while reproducing, generalizing or improving state-of-the-art methods.
Starting with a controlled $W\in SU(2)$ ($CW$), we can implement $W=\rv{\lambda}{\hat{v}}$ as $\rpisub{2}\rpisub{1}$ with $\hat{v_1},\hat{v_2}$ perpendicular to $\hat{v}$, and $\ang{\hat{v_1}}{\hat{v_2}}=\tfrac{\lambda}{2}$
using \lem{2_rpi}. Adding a control to each of the $\Pi$ gates only requires a single CNOT according to \thm{1CNOT_iff}. \qc{cu_rpi} therefore requires two CNOT gates, similarly to the well known decomposition from \cite{barenco_elementary_1995}. Clearly, any controlled $U\in U(2)$ ($CU$) gate can be implemented by applying a phase gate on the control qubit.
\[
\scalebox{0.8}{
\Qcircuit @C=1.0em @R=0.3em @!R {
	 	\nghost{{c} :  } & \lstick{{c} :  } & \ctrl{1} & \qw \\
	 	\nghost{{t} :  } & \lstick{{t} :  } & \gate{\mathrm{W}} & \qw \\
}
\hspace{5mm}\raisebox{-3.5mm}{=}\hspace{0mm}
\Qcircuit @C=0.2em @R=0.2em @!R { 
	 	\nghost{{c :} :  } & \lstick{{c :}} & \ctrl{1} & \ctrl{1} & \qw \\
	 	\nghost{{t :} :  } & \lstick{{t :}} & \rpigatesub{1} & \rpigatesub{2} & \qw \\
 }
 }
 \qcref{cu_rpi}
\]
In order to add another control, we can use 
\lem{sq_trans} to apply an axis transformation and implement $\rpisub{2}$ as $\rpisub{3}\rpisub{1}\rpisub{3}$, defined by $\hat{v}_3$ on the same plane as $\hat{v}_1,\hat{v}_2$, and satisfying $\ang{\hat{v_1}}{\hat{v_3}}=\ang{\hat{v_3}}{\hat{v_2}}=\tfrac{\lambda}{4}$. As can be seen in \qc{mcrpi_base}, qubit $c_1$ controls the axis transformation, such that for the $\ket{0}_{c_1}$ state, the controlled-$\rpisub{1}$ gates cancel out since they are Hermitian, and for the $\ket{1}_{c_1}$ state, \qc{cu_rpi} is applied on $c_2,t$. 
\[
\scalebox{0.8}{
\Qcircuit @C=1.0em @R=0.1em @!R { 
	 	\nghost{{c_1 :} :  } & \lstick{{c_1 :}}  & \ctrl{1} & \qw \\
	 	\nghost{{c_2 :} :  } & \lstick{{c_2 :}}  & \ctrl{1} & \qw \\
	 	\nghost{{t :} :  } & \lstick{{t :}} & \gate{\mathrm{W}} & \qw \\
 }
\hspace{5mm}\raisebox{-6mm}{=}\hspace{0mm}
\Qcircuit @C=0.2em @R=0.0em @!R { 
	 	\nghost{{c_1 :} :  } & \lstick{{c_1 :}} & \qw & \ctrl{2} & \qw & \ctrl{2} & \qw  \\
	 	\nghost{{c_2 :} :  } & \lstick{{c_2 :}} & \ctrl{1} & \qw & \ctrl{1} & \qw  & \qw \\
	 	\nghost{{t :} :  } & \lstick{{t :}} & \rpigatesub{1} & \rpigatesub{3} & \rpigatesub{1} & \rpigatesub{3} & \qw \\
 }
 }
 \qcref{mcrpi_base}
\]
By adding a controlled phase gate on $c_1,c_2$ this structure can implement {\em any} controlled-controlled $U(2)$ ($C^2U$) gate using six CNOTs, i.e. at the same cost as a Toffoli.
Now, we wish to implement the $C^2U$ gate in linear nearest-neighbor (LNN) connectivity. The control condition can be rewritten as $\{(c_1 \oplus c_2),c_2\}=\{0,1\}$ instead of $\{c_1,c_2\}=\{1,1\}$. We would like to use $c_2$ and $c_1 \oplus c_2$ to control our $\Pi$ gates, in order to remove the need to swap $c_1$ towards the target. We introduce the gate $\MCO{\Pi}{\{c_1 \oplus c_2\}}{t} = \MCO{X}{c_1}{c_2}\MCO{\Pi}{c_2}{t}\MCO{X}{c_1}{c_2}$ for this task. This time, we use the axis transformation to replace $\rpisub{2}$ with $\rpisub{3}\rpisub{2}\rpisub{3}$. In this case, $\rpisub{3}$ is controlled by $c_1 \oplus c_2$ and will change $\rpisub{2}$ to $\rpisub{1}$ only if $c_1 \oplus c_2 = 1$. Therefore, $\MCO{\rpisub{3}}{\{c_1 \oplus c_2\}}{t}\MCO{\rpisub{2}}{c_2}{t}\MCO{\rpisub{3}}{\{c_1 \oplus c_2\}}{t}\MCO{\rpisub{1}}{c_2}{t}$ implements the $C^2W$ gate. This can be used to implement any $C^2U$ gate by adding a controlled phase gate as $\MCO{\Pi_4}{c_1}{c_2}\sqO{P}{c_1}\MCO{X}{c_1}{c_2}$, with some $\hat{v}_4\perp\hat{z}$ and $P$ a single-qubit phase gate. Noting that two CNOT gates cancel out, we reach the 8 CNOT implementation in \qc{ccu_lnn}, again providing the same CNOT count as a Toffoli \cite{nemkov_efficient_2023,gwinner_benchmarking_2021,nakanishi_quantum-gate_2021,nakanishi_decompositions_2024,m_q_cruz_shallow_2024}, while generalizing the controlled operator to {\em any} U(2) for a specific target location.
\[
\scalebox{0.8}{
\Qcircuit @C=1.0em @R=0.1em @!R { 
	 	\nghost{{c_1 :} :  } & \lstick{{c_1 :}} & \ctrl{1} & \qw \\
	 	\nghost{{c_2 :} :  } & \lstick{{c_2 :}} & \ctrl{1} & \qw \\
	 	\nghost{{t :} :  } & \lstick{{t :}}  & \gate{\mathrm{U}} & \qw \\
 }
\hspace{5mm}\raisebox{-6mm}{=}\hspace{0mm}
\Qcircuit @C=0.1em @R=0.0em @!R { 
	 	\nghost{{c_1 :} :  } & \lstick{{c_1 :}} & \qw & \ctrl{1} & \qw & \ctrl{1} & \qw & \ctrl{1} & \gate{\mathrm{P}} & \ctrl{1} & \qw \\
	 	\nghost{{c_2 :} :  } & \lstick{{c_2 :}}  & \ctrl{1} & \targ & \ctrl{1} & \targ & \ctrl{1} & \targ & \ctrl{1} & \rpigatesub{4} & \qw \\
	 	\nghost{{t :} :  } & \lstick{{t :}} & \rpigatesub{1} & \qw & \rpigatesub{3} & \qw & \rpigatesub{2} & \qw & \rpigatesub{3} & \qw & \qw \\
 }
 }
 \qcref{ccu_lnn}
\]
In case \qc{ccu_lnn} is used in order to implement an $R_x(\lambda)$ gate with two controls, we can always choose $\rpisub{3}=Z$, and define $\hat{v}_1,\hat{v}_2$ accordingly on the $zy$ plane as $\hat{v}(\tfrac{\lambda}{4},\tfrac{\pi}{2})$ and $\hat{v}(\tfrac{\lambda}{4},-\tfrac{\pi}{2})$ respectively. Moreover, since $R_x\in SU(2)$, we get $\rpisub{4}=X$, and the $P$ gate can be removed.  Using \lem{MC_transform}, we can decompose $\MCO{\rpisub{1}}{c_2}{t} = \sqO{\rpisub{5}}{t}\MCO{Z}{c_2}{t}\sqO{\rpisub{5}}{t}$, and $\MCO{\rpisub{2}}{c_2}{t} = \sqO{\rpisub{6}}{t}\MCO{Z}{c_2}{t}\sqO{\rpisub{6}}{t}$ with $\hat{v}_5=\hat{v}(\tfrac{\lambda}{8},\tfrac{\pi}{2})$ and $\hat{v}_6=\hat{v}(\tfrac{\lambda}{8},-\tfrac{\pi}{2})$. By adding a control to the axis transformation gates $\rpisub{5},\rpisub{6}$ we achieve \qc{cccrx_lnn}. We note that these controlled $\Pi$ gates can be decomposed using one CZ and two $R_x$ gates using \lem{Cpi_zz_yy}. 
\[
 \resizebox{0.9\linewidth}{!}{
\Qcircuit @C=1.0em @R=0em @!R { 
	 	\nghost{{c_1 :} :  } & \lstick{{c_1 :}} & \ctrl{1} & \qw \\
	 	\nghost{{c_2 :} :  } & \lstick{{c_2 :}} & \ctrl{1} & \qw \\
	 	\nghost{{t :} :  } & \lstick{{t :}} & \gate{\mathrm{R_x}} & \qw \\
	 	\nghost{{c_3 :} :  } & \lstick{{c_3 :}}& \ctrl{-1} & \qw \\
 }
 \hspace{5mm}\raisebox{-8mm}{=}\hspace{0mm}
\Qcircuit @C=0.4em @R=0.0em @!R { 
	 	\nghost{{c_1 :} :  } & \lstick{{c_1 :}} & \qw & \qw & \ctrl{1} & \qw & \ctrl{1} & \qw & \ctrl{1} & \qw & \ctrl{1} & \qw \\
	 	\nghost{{c_2 :} :  } & \lstick{{c_2 :}} & \qw & \ctrl{1} & \targ & \ctrl{1} & \targ & \ctrl{1} & \targ & \ctrl{1} & \targ & \qw \\
	 	\nghost{{t :} :  } & \lstick{{t :}} & \rpigatesub{5} & \control\qw & \rpigatesub{5} & \control\qw & \rpigatesub{6} & \control\qw & \rpigatesub{6} & \control\qw & \qw & \qw \\
	 	\nghost{{c_3 :} :  } & \lstick{{c_3 :}} & \ctrl{-1} & \qw & \ctrl{-1} & \qw & \ctrl{-1} & \qw & \ctrl{-1} & \qw & \qw & \qw \\
 }
 }
 \qcref{cccrx_lnn}
\]
It can be noticed that the rightmost $\MCO{X}{c_1}{c_2}\MCO{Z}{c_2}{t}\MCO{X}{c_1}{c_2}$ is equivalent to $\MCO{Z}{c_1}{t}\MCO{Z}{c_2}{t}$ \cite{maslov_depth_2022}. 
We can remove these gates when implementing a $C^3R_x$ up-to a relative phase, which can therefore be implemented using 9 CNOT gates. In case the target is located at the bottom, the gate requires 11 CNOT gates as \qc{cccrx_lnn_bot}. Since any $R_x$ can be decomposed as $Z\Pi_{\bar{x}}$, for $\Pi_{\bar{x}}$ defined according to \lem{2_rpi} as a $\pi$-rotation about an axis on the $zy$ plane, we note that \qc{cccrx_lnn_bot} implements any chosen $C^3\Pi_{\bar{x}}$ up to a relative phase as well.
\[
\scalebox{0.8}{
\Qcircuit @C=0.2em @R=0em @!R { 
	 	\lstick{} & \ctrl{1} & \multigate{3}{\mathrm{\Delta'}} & \qw \\
	 	\lstick{} & \ctrl{1} & \ghost{\mathrm{\Delta'}} & \qw \\
        \lstick{} & \ctrl{1} &\ghost{\mathrm{\Delta'}} & \qw \\
	 	\lstick{} & \gate{\mathrm{R_x}} & \ghost{\mathrm{\Delta'}} & \qw \\
 }
 \hspace{2mm}\raisebox{-7.4mm}{=}\hspace{2mm}
 \Qcircuit @C=0.2em @R=0em @!R { 
	 	\lstick{} & \ctrl{1} & \multigate{3}{\mathrm{\Delta}} & \qw \\
	 	\lstick{} & \ctrl{1} & \ghost{\mathrm{\Delta}} & \qw \\
        \lstick{} & \ctrl{1} &\ghost{\mathrm{\Delta}} & \qw \\
	 	\lstick{} & \gate{\mathrm{\Pi_{\bar{x}}}} & \ghost{\mathrm{\Delta}} & \qw \\
 }
 \hspace{2mm}\raisebox{-7mm}{=}\hspace{2mm}
\Qcircuit @C=0.4em @R=0.0em @!R { 
	 	\lstick{} & \qw & \qw & \qw & \ctrl{1} & \qw & \ctrl{1} & \qw & \qw & \qw & \qw \\
	 	\lstick{} & \qw & \qw & \ctrl{1} & \targ & \ctrl{1} & \targ & \ctrl{1} & \qw & \qw & \qw  \\
	 	\lstick{} & \qswap & \rpigatesub{5} & \control\qw & \rpigatesub{5} & \control\qw & \rpigatesub{6} & \control\qw & \rpigatesub{6} & \qswap & \qw  \\
	 	\lstick{} &\qswap \qwx[-1] & \ctrl{-1} & \qw & \ctrl{-1} & \qw & \ctrl{-1} & \qw & \ctrl{-1} & \qswap \qwx[-1] & \qw 
        \gategroup{3}{2}{4}{3}{0.5em}{--} 
        \gategroup{3}{9}{4}{10}{0.5em}{--} 
 }
 }
 \qcref{cccrx_lnn_bot}
\]
with $\Delta,\Delta'$ as the relative phase (diagonal) gates. The boxed gates only require two CNOTs each as \qc{SWAP_CPi}.
\[
\scalebox{0.8}{
\Qcircuit @C=1.0em @R=0.3em @!R {
	 	\lstick{} & \qswap & \rpigatesub{} & \qw \\
	 	\lstick{} & \qswap \qwx[-1] & \ctrl{-1} & \qw \\
}
\hspace{2mm}\raisebox{-3.5mm}{=}\hspace{2mm}
\Qcircuit @C=0.5em @R=0.3em @!R {
	 	\lstick{} & \qw & \qswap & \targ & \gate{\mathrm{\Pi'}} & \qw \\
	 	\lstick{} & \gate{\mathrm{\Pi'}} & \qswap \qwx[-1] & \ctrl{-1} & \qw  & \qw \\
}
\hspace{2mm}\raisebox{-3.5mm}{=}\hspace{2mm}
\Qcircuit @C=0.2em @R=0.3em @!R {
	 	\lstick{} & \qw & \targ & \ctrl{1} & \gate{\mathrm{\Pi'}} & \qw \\
	 	\lstick{} & \gate{\mathrm{\Pi'}} & \ctrl{-1} & \targ & \qw  & \qw \\
}
 }
 \qcref{SWAP_CPi}
\]

Noting that by choosing the target rotation as $R_x(\pi)$, or equivalently $\rpisub{\bar{x}}=Y$, \qc{cccrx_lnn_bot} implements a 3-controlled Toffoli up-to a relative phase ($C^3X\text{-}\Delta$) at the same CNOT count reported by \cite{nemkov_efficient_2023}.
Our generalization of the $C^3X\text{-}\Delta$ gate to $C^3\rpisub{\bar{x}}\text{-}\Delta$  proves to be useful in the construction of the 4-controlled Toffoli ($C^4X$). 
First, we implement the $C^4R_x$ gate as \qc{ccccrx_lnn}. It can be noted that the relative phase gates $\Delta,\Delta^\dagger$ can commute with the CZ gate between them and cancel out. Ignoring these gates, this circuit is similar to \qc{mcrpi_base}, with an axis transformation controlled by three qubits. Specifically, in this case $\Pi_1$ is replaced with $Z$, and $\Pi_3$ with $\Pi_{\bar{x}}$, thus allowing to implement any four-controlled $R_x$ gate.
\[
\scalebox{0.8}{
\Qcircuit @C=1.0em @R=0em @!R { 
	 	\lstick{} & \ctrl{1}  & \qw \\
	 	\lstick{} & \ctrl{1}  & \qw \\
        \lstick{} & \ctrl{1}  & \qw \\
	 	\lstick{} & \gate{\mathrm{R_x}} &  \qw \\
        \lstick{} & \ctrl{-1}  & \qw \\
 }
 \hspace{2mm}\raisebox{-12mm}{=}\hspace{2mm}
\Qcircuit @C=1.0em @R=0em @!R { 
	 	\lstick{} & \qw & \ctrl{1} & \multigate{3}{\mathrm{\Delta}} & \qw & \multigate{3}{\mathrm{\Delta^\dagger}} & \ctrl{1}  & \qw \\
	 	\lstick{} & \qw & \ctrl{1} & \ghost{\mathrm{\Delta}} & \qw & \ghost{\mathrm{\Delta^\dagger}} & \ctrl{1}  & \qw \\
        \lstick{} & \qw  & \ctrl{1} &\ghost{\mathrm{\Delta}} & \qw & \ghost{\mathrm{\Delta^\dagger}} & \ctrl{1} & \qw \\
	 	\lstick{} & \control\qw & \gate{\mathrm{\Pi_{\bar{x}}}} & \ghost{\mathrm{\Delta}} & \control\qw & \ghost{\mathrm{\Delta^\dagger}} & \gate{\mathrm{\Pi_{\bar{x}}}}  & \qw \\
        \lstick{} & \ctrl{-1} & \qw & \qw & \ctrl{-1} & \qw & \qw  & \qw
        \gategroup{1}{3}{4}{4}{0.3em}{--} 
        \gategroup{1}{6}{4}{7}{0.3em}{--} 
 }
 }
 \qcref{ccccrx_lnn}
\]
Noting that this gate can be used for any chosen $R_x$ rotation, we can choose $R_x(\pi)=-iX$. The $-i$ phase can be removed by applying a $C^3S$ gate on the control qubits. We implement the $C^4X$ gate as \qc{ccccx_lnn}, applying SWAP gates to place the target of our $C^4R_x$ gate at the bottom. Each SWAP gate only requires two CNOT gates instead of three, since a CNOT gate applied on the target of the $C^4R_x$ gate commutes with it.
\[
\scalebox{0.8}{
\Qcircuit @C=1.0em @R=0.8em @!R { 
	 	\lstick{} & \ctrl{1}  & \qw \\
	 	\lstick{} & \ctrl{1}  & \qw \\
        \lstick{} & \ctrl{1}  & \qw \\
        \lstick{} & \ctrl{1}  & \qw \\
	 	\lstick{} & \targ &  \qw \\
 }
 \hspace{2mm}\raisebox{-13mm}{=}\hspace{2mm}
 \Qcircuit @C=0.5em @R=0em @!R { 
	 	\lstick{} & \ctrl{1} & \ctrl{1} & \qw\\
	 	\lstick{} & \ctrl{1}  & \ctrl{1} & \qw\\
        \lstick{}  & \ctrl{1}  & \ctrl{1} & \qw\\
	 	\lstick{}  & \ctrl{1} & \gate{\mathrm{S}} &  \qw \\
        \lstick{}  & \gate{\mathrm{R_x(\pi)}}   & \qw & \qw\\
 }
 \hspace{2mm}\raisebox{-13mm}{=}\hspace{2mm}
 \Qcircuit @C=0.2em @R=0em @!R { 
	 	\lstick{}&  \qw  & \qw & \ctrl{1} &  \qw  & \qw & \ctrl{1} & \qw\\
	 	\lstick{}&  \qw  & \qw & \ctrl{1} &  \qw  & \qw & \ctrl{1} & \qw\\
        \lstick{}&  \qw  & \qw & \ctrl{1}&  \qw   & \qw & \ctrl{1} & \qw\\
	 	\lstick{} & \targ & \ctrl{1} & \gate{\mathrm{R_x(\pi)}} & \ctrl{1} & \targ& \gate{\mathrm{S}} &  \qw \\
        \lstick{} & \ctrl{-1} & \targ & \ctrl{-1} & \targ & \ctrl{-1}  & \qw & \qw\\
 }
 }
 \qcref{ccccx_lnn}
\]
Finally, noting that the leftmost CZ in \qc{ccccrx_lnn}, and one CNOT gate from \qc{ccccx_lnn} can be merged to one controlled $\pi$-rotation as $\MCO{Z}{c}{t}\MCO{X}{c}{t} = \MCO{iY}{c}{t}$, these can be implemented using one CNOT gate. We get that the $C^4R_x$ gate with a target at the bottom requires $27$ CNOT gates, as it uses two $C^3\Pi_{\bar{x}}\text{-}\Delta$ gates and five $C\Pi$ gates. Since the $C^3S$ gate can be implemented using 18 CNOT gates as a $C^3V$ \cite{nemkov_efficient_2023} with two Hadamard gates, we managed to reduce the cost of the $C^4X$ (aka the 5q Toffoli) gate in LNN connectivity from 48 \cite{nemkov_efficient_2023,allende_synthesis_2024} to 45. As the best known decomposition of the $C^4X$ gate without ancilla in unrestricted connectivity requires $30$ CNOT gates \cite{nemkov_efficient_2023,shende_synthesis_2006}, we effectively reduce the CNOT overhead required to to map this gate to LNN connectivity by $16.6\%$.}
{We show in \cite{zindorf_efficient_2025,zindorf_multi-controlled_2025} that these methods result in reduced CNOT count for any number of controls, both for LNN and for unrestricted connectivity. }
\\

\section{Controlled U(4) gates}\label{sec:Controlled_U4}\label{ssec:implementations} 
    In this section we focus on the implementation of controlled two-qubit operators $\MCO{V}{c}{\{t_1,t_2\}}$ with any $V\in U(4)$. Our goal is to minimize the number of CNOT gates used. The approach we take demonstrates the ability to "add a control" to a given circuit (in this case an arbitrary two-qubit gate) by converting single qubit rotations to Hermitian gates, and adding a control to these gates at a reduced cost using \thm{1CNOT_iff}.
First we show that the following structure can implement any $V\in U(4)$ gate up to a global phase:
\[
 \resizebox{0.85\linewidth}{!}{
\Qcircuit @C=1.0em @R=0.2em @!R { 
	 	\nghost{{t}_{1} :  } & \lstick{{t}_{1} :  } & \rpigatesubu{z}{1} & \gate{\mathrm{U_1}} & \ctrl{1} & \rpigatesubu{y}{3} & \ctrl{0} & \rpigatesubu{y}{5} & \ctrl{1} & \gate{\mathrm{U^\dagger_1}} & \rpigatesub{6}& \qw \\
	 	\nghost{{t}_{2} :  } & \lstick{{t}_{2} :  } & \rpigatesubu{z}{2} & \gate{\mathrm{U_2}} & \gate{\mathrm{iY}}  & \rpigatesubu{z}{4} & \ctrl{-1} & \qw & \targ & \gate{\mathrm{U^\dagger_2}} & \rpigatesub{7} & \qw \\
 }
 }
 \qcref{su4_ours}
\]
With $U_1,U_2\in SU(2)$, and $\rpi{\hat{v}_{j\in\{1,2,..,7\}}}$ gates are marked as $\rpisubu{j}{\sigma}$ if it is known that $\hat{v}_j \perp \hat{\sigma}$ or as $\rpisub{j}$ otherwise. We note that the $\rpisubu{j}{\sigma}$ gates require one parameter, the $\rpisub{j}$ require two parameters, and each $SU(2)$ gate and its inverse use the same three parameters, therefore this circuit uses $15$ parameters which is the minimal requirement for a $U(4)$ operator disregarding the global phase.
\begin{lemma}\label{lem:SU4_pi_rot}
    Any $V\in U(4)$ gate can be implemented up to a global phase using \qc{su4_ours}, such that $\hat{v}_3,\hat{v}_5$ are on the $xz$ plane, and $\hat{v}_1,\hat{v}_2,\hat{v}_4$ are on the $xy$ plane.
\end{lemma}
\begin{proof}
The following circuit is known to implement any $V\in U(4)$ up to a global phase:
\[
\scalebox{0.8}{
\Qcircuit @C=1.0em @R=0.2em @!R { 
	 	\nghost{{t}_{1} :  } & \lstick{{t}_{1} :  } & \gate{\mathrm{D}} & \ctrl{1} & \rvgate{\theta_1}{\hat{y}} & \targ & \rvgate{\theta_2}{\hat{y}} & \ctrl{1} & \gate{\mathrm{B}} & \qw \\
	 	\nghost{{t}_{2} :  } & \lstick{{t}_{2} :  } & \gate{\mathrm{C}} & \targ & \rvgate{\theta_3}{\hat{z}} & \ctrl{-1} & \qw & \targ & \gate{\mathrm{A}} & \qw \\
 }
 }
 \qcref{su4_gen}
\]
where the angles $\theta_1,\theta_2,\theta_3$, and the operators $A,B,C,D\in SU(2)$ can be found as shown in \cite{shende_minimal_2004,vatan_optimal_2004,vidal_universal_2004} for any chosen $V$.
The middle CNOT gate can be replaced by $\sqO{H}{t_1}\MCO{Z}{t_2}{t_1}\sqO{H}{t_1}$, and the first CNOT by $\sqO{S^\dagger}{t_1}\sqO{\rpisub{S}}{t_2}\MCO{iY}{t_1}{t_2}\sqO{\rpisub{S}}{t_2}$ from \lem{MC_transform}.
The following circuit therefore implements $V$ up to a global phase:
\[
 \resizebox{0.85\linewidth}{!}{
\Qcircuit @C=0.5em @R=0.2em @!R { 
	 	\nghost{{t}_{1} :  } & \lstick{{t}_{1} :  } & \gate{\mathrm{D'}}  & \ctrl{1} & \rvgate{\theta_1}{\hat{y}} &  \gate{\mathrm{H}} & \ctrl{0} &  \gate{\mathrm{H}} & \rvgate{\theta_2}{\hat{y}} & \ctrl{1} & \gate{\mathrm{B}} & \qw \\
	 	\nghost{{t}_{2} :  } & \lstick{{t}_{2} :  } & \gate{\mathrm{C'}}  & \gate{\mathrm{iY}} & \rpigatesub{S} & \rvgate{\theta_3}{\hat{z}} & \ctrl{-1} & \qw & \qw & \targ & \gate{\mathrm{A}} & \qw \\
 }
 }
 \qcref{su4_gen2}
\]
with $D':=S^\dagger D$, and $C':=\rpisub{S}C$.
The gate $\rv{\theta_3}{\hat{z}}$ can be decomposed as $\rpi{\hat{v}_4}\rpisub{S}$ with $\hat{v}_4 = \hat{v}(\frac{\pi}{2},\frac{\pi}{4}+\frac{\theta_3}{2})$ using
\lem{2_rpi}.
Similarly, $H\rv{\theta_1}{\hat{y}} = \rpi{\hat{v}_3}$, and $\rv{\theta_2}{\hat{y}}H = \rpi{\hat{v}_5}$ with $\hat{v}_3=\hat{R}_{\hat{y}}(-\frac{\theta_1}{2})\hat{h}$, $\hat{v}_5=\hat{R}_{\hat{y}}(\frac{\theta_2}{2})\hat{h}$ such that $\rpi{\hat{h}}=H$  for $\hat{h} = \hat{v}(\frac{\pi}{4},0)$. We note that $\hat{v}_3,\hat{v}_5$ are on the $xz$ plane, and $\hat{v}_4$ is on the $xy$ plane. \\
Finally, we define $U_{BD}:=BD'$. This operator can be expressed as $U_{BD} = \rv{\theta_{B}}{\hat{v}_{B}}$ up to a global phase with a unit vector $\hat{v}_B$ and an angle $\theta_B$. $U_{BD}$ can therefore be decomposed using \lem{2_rpi} as $\rpi{\hat{v}_6}\rpi{\hat{v}_1}$. The vector $\hat{v}_1$ can be chosen to be perpendicular to both $\hat{v}_{B}$ and $\hat{z}$ such that it is on the $xy$ plane, and  $\hat{v}_6=\hat{R}_{\hat{v}_B}(\frac{\theta_B}{2})\hat{v}_1$. Since $\rpicl$ gates are Hermitian we can write $\rpi{\hat{v}_6}BD'\rpi{\hat{v}_1}=I$. 
Defining $U_1=D'\rpi{\hat{v}_1}$ up to a global phase, such that $U_1\in SU(2)$, therefore provides $U^\dagger_1=\rpi{\hat{v}_6}B$. Repeating this for $A,C'$ produces $U_2,U^\dagger_2,\hat{v}_2,\hat{v}_7$ such that $\hat{v}_2$ is on the $xy$ plane.
\end{proof}

\qc{CU4_a2a} is achieved by replacing each $\Pi$ gate in \qc{su4_ours} with its controlled version.
\[
 \resizebox{0.85\linewidth}{!}{
\Qcircuit @C=0.8em @R=0.15em @!R { 
	 	\nghost{{c} :  } & \lstick{{c} :  } & \ctrl{1}  & \gate{\mathrm{P(\phi)}} & \qw & \ctrl{1} & \qw & \ctrl{1} & \qw & \qw & \ctrl{1}& \qw\\
	 	\nghost{{t}_{1} :  } & \lstick{{t}_{1} :  } & \rpigatesubu{z}{1}\ar @{-} [1,0] & \gate{\mathrm{U_1}} & \ctrl{1} & \rpigatesubu{y}{3}\ar @{-} [1,0] & \ctrl{0} & \rpigatesubu{y}{5} & \ctrl{1} & \gate{\mathrm{U^\dagger_1}} & \rpigatesub{6}\ar @{-} [1,0] & \qw \\
	 	\nghost{{t}_{2} :  } & \lstick{{t}_{2} :  } & \rpigatesubu{z}{2} & \gate{\mathrm{U_2}} & \gate{\mathrm{iY}}  & \rpigatesubu{z}{4} & \ctrl{-1} & \qw & \targ & \gate{\mathrm{U^\dagger_2}} & \rpigatesub{7} & \qw \\
 }
 }
 \qcref{CU4_a2a}
\]

 We show in \apx{controlled_u4_apx} that this circuit, which requires ten controlled $\pi$-rotations, can implement any $\MCO{V}{c}{\{t_1,t_2\}}$ gate. Circuits for LNN restricted
connectivity are provided as well.
In addition, we present methods to efficiently convert a circuit which was designed using controlled $\pi$-rotation gates to other sets of universal gates. These methods are then used to decompose our Controlled $U(4)$ implementation to CNOT gates, along with arbitrary rotations, $\pi$ rotations or $\{R_{\hat{z}},R_{\hat{y}}\}$ gates.
The results are summarized in \tab{csu4_costs}.

\begin{table}[h]
    \centering
    \begin{tabular}{ |p{2cm}||p{1.1cm}|p{1.7cm}|  }
 \hline
 Connectivity&CNOT &$R_{\hat{v}}$/$\rpicl$/$R_{\hat{z},\hat{y}}$ \\
 \hline
 All-to-all  & 10   & 15/20/25 \\
 LNN             & 13   & 15/20/25 \\
 \hline    \end{tabular}   \caption{Counts of CNOT gates and single-qubit gates from three universal gate sets required for our implementation of any Controlled-$U(4)$ gate in all-to-all or LNN connectivity.}
\label{tab:csu4_costs}
\end{table}
\section{Conclusion}
{In this paper, we explored the possibility and benefits of doing quantum computation
purely by using single and two qubit Hermitian gates.} 
Specifically, $\pi$-rotations about two fixed axes, along with the CNOT gate are universal for quantum computation. When arbitrary axes are available, any single qubit gate can be decomposed as two $\pi$-rotations.
{As shown here, this decomposition is particularly beneficial for single-qubit state preparation in the presence of systematic amplitude errors, while an alternative sequence of four $\pi$-rotations allows to improve the fidelity of arbitrary single-qubit gates.}

In addition, we have shown that two $\Pi$ gates can be used in order to transform the rotation axis of multi-controlled $U(2)$ gates. It follows that the CNOT count of an $n$-controlled $SU(2)$ gate is independent of the choice of rotation axis. Therefore, the CNOT count of $20n$ reported by \cite{vale_circuit_2024} immediately reduces to $16n$, which has been shown to be sufficient for specific chosen axes by the same paper.
This axis transformation also allows to implement any singly-controlled $\Pi$ gate using one CNOT gate.
 
The universality of the controlled $\Pi$ gates{, along with their reduced CNOT count,} provides a convenient framework for circuit optimization and compilation. We have demonstrated some benefits of this approach while optimizing controlled $U(4)$ gates, resulting in a low number of CNOT gates for both all-to-all and LNN architectures. In the process, we have provided methods that allow one to efficiently convert a circuit which has been constructed using controlled $\Pi$ gates to another set of gates.
{In addition, we used this framework to implement various multi-controlled single-target gates with a small number of controls, while either reproducing or generalizing state-of-the-art methods. For the case of the 5q Toffoli gate without ancilla in LNN connectivity, we managed to reduce the CNOT count from $48$ \cite{nemkov_efficient_2023} to $45$. This reduction becomes especially advantageous in circuits which use many instances of such gates.
Additional cases in which the use of Hermitian gates can provide efficient implementations of quantum operators have been covered in \cite{zindorf_efficient_2025,zindorf_multi-controlled_2025}.}

Focusing on Hermitian gates may show benefits when implemented directly as $\pi$ rotations using composite pulses, in addition to providing convenient tools for circuit optimization tasks. 

{For future research, one may consider the situation in which the strength of two qubit interactions can be so adjusted that the entangling part/s of a CNOT require the same time as a $\pi$ rotation on a single qubit. Under these circumstances, as local gates can always be accomplished as $\pi$ rotations, every operation can be accomplished by repeating the same pulse. Thus, $\pi$ pulses applied by a control machine may be applied as a square wave with constant frequency, similarly to a clock used for classical computation. This can potentially make it easier to establish a clocked version of quantum computation.
}
\\
\\
\textbf{Acknowledgements}\\
This work was supported by the Engineering and Physical Sciences Research Council [grant numbers EP/R513143/1, EP/T517793/1].

\newpage
\appendix

\section*{Appendix}\label{app:appendix}
\section{Controlled U(4) gates construction}\label{apx:controlled_u4_apx}
We provide a constructive proof for our controlled-$U(4)$ structure which is achieved by 'adding a control' to \qc{su4_ours}.
\begin{lemma}\label{lem:is_SU4}
    Any operator which can be implemented exactly using \qc{su4_ours} is in $SU(4)$.
\end{lemma}
\begin{proof} 
    Any circuit over two qubits which is composed solely of $SU(2)$ and $SU(4)$ operator implements an $SU(4)$ operator. 
    Therefore, we simply express each gate in \qc{su4_ours} as a $SU(2)$ or $SU(4)$ operator multiplied by a phase, and show that the total added phase equals to $1$. By definition, $\rpi{\hat{v}}=i\rv{\pi}{\hat{v}}$ such that $\rv{\pi}{\hat{v}}\in SU(2)$. The $\MCO{iY}{t_1}{t_2}$ gate is in $SU(4)$ since its determinant equals $1$. The CNOT and CZ gates have a determinant of $-1$, and therefore can be written as $\MCO{X}{t_1}{t_2}=\sqrt{i}G$ and $\MCO{Z}{t_1}{t_2}=\sqrt{i}F$ with $G,F\in SU(4)$. The rest of the gates are $SU(2)$ operators, and therefore, the total added phase is $i^7\sqrt{i}\sqrt{i}=1$.
\end{proof}

\begin{lemma}\label{lem:CU4_a2a_lem}
Any $\MCO{V}{c}{\{t_1,t_2\}}$ gate can be implemented with 10 CNOT gates using \qc{CU4_a2a}, s.t. $\hat{v}_3,\hat{v}_5$ are on the $xz$ plane, $\hat{v}_1,\hat{v}_2,\hat{v}_4$ are on the $xy$ plane, and $V\in U(4)$.
\end{lemma}
\begin{proof}
According to 
\lem{SU4_pi_rot} and \lem{is_SU4}, for any operator $V\in U(4)$, there exists an operator $U\in SU(4)$ which can be implemented exactly using \qc{su4_ours} such that $V=e^{i\phi}U$. Since the determinant of any $SU(4)$ operator equals $1$, the angle $\phi\in [0,2\pi)$ must satisfy $e^{i4\phi}=det(V)$. Thus, there are four options for the value of $\phi$, and it will be chosen as one which allows $U$ to be implemented exactly using \qc{su4_ours}.
The parameters defining $U_1,U_2\in SU(2)$, and $\hat{v}_{j\in[1,7]}$ can be found as shown in \lem{SU4_pi_rot} such that $\hat{v}_3,\hat{v}_5$ are on the $xz$ plane, $\hat{v}_1,\hat{v}_2,\hat{v}_4$ are on the $xy$ plane. We investigate the operator applied on $t_1,t_2$ by \qc{CU4_a2a} for each computational basis state of the control qubit $c$.

For the state $\ket{1}$ of the control qubit $c$, the operator applied on qubits $t_1,t_2$ is described by \qc{su4_ours}, multiplied by a phase $e^{i\phi}$, and is therefore $e^{i\phi}U=V$.

For the state $\ket{0}$ of the control qubit $c$, the operator applied on qubits $t_1,t_2$ by \qc{CU4_a2a} is as follows:
\[
 \resizebox{0.9\linewidth}{!}{
\Qcircuit @C=1.0em @R=0.2em @!R { 
	 	\nghost{{t}_{1} :  } & \lstick{{t}_{1} :  } & \gate{\mathrm{U_1}} & \ctrl{1}  & \ctrl{0}  & \ctrl{1} & \gate{\mathrm{U^\dagger_1}} & \qw \\
	 	\nghost{{t}_{2} :  } & \lstick{{t}_{2} :  } & \gate{\mathrm{U_2}} & \gate{\mathrm{iY}} & \ctrl{-1} & \targ  & \gate{\mathrm{U^\dagger_2}} & \qw \\
 }
\hspace{5mm}\raisebox{-5mm}{=}\hspace{0mm}
\Qcircuit @C=1.0em @R=0.2em @!R { 
	 	\nghost{{t}_{1} :  } & \lstick{{t}_{1} :  } & \gate{\mathrm{U_1}} & \ctrl{1}  & \gate{\mathrm{U^\dagger_1}} & \qw \\
	 	\nghost{{t}_{2} :  } & \lstick{{t}_{2} :  } & \gate{\mathrm{U_2}} & \gate{\mathrm{iXZY}} & \gate{\mathrm{U^\dagger_2}} & \qw \\
 }
 }
\]
Since $iXZY=I$, and ${U_1}^\dagger U_1={U_2}^\dagger U_2=I$ , this circuit implements the $4\times 4$ identity matrix.\\
Therefore, \qc{CU4_a2a} can implement any controlled $V\in U(4)$ operator. The circuit requires ten $\MCO{e^{i\psi}\rpicl}{c}{t}$ gates, each requires one CNOT according to \thm{1CNOT_iff}.
\end{proof}

We now provide implementations of the controlled $U(4)$ gate for LNN connectivity.

\begin{lemma}\label{lem:CU4_LNN_top}
    Any $\MCO{V}{c}{\{t_1,t_2\}}$ gate with $V\in U(4)$ can be implemented in LNN architecture over 3 neighbouring qubits using 13 CNOT gates if $c$ is the first or the last qubit.
\end{lemma}
\begin{proof}
According to \lem{CU4_a2a_lem}, any $\MCO{V}{c}{\{t_1,t_2\}}$ operator can be implemented using the structure described in \qc{CU4_a2a} if $c$ is the first qubit (if it is the last qubit then the structure is simply flipped). This structure requires three gates of the form $\MCO{\rpisub{b}}{c}{t_2}\MCO{\rpisub{a}}{c}{t_1}$. Each of these can be decomposed in LNN connectivity using 3 CNOT gates simply by decomposing the gates as \qc{crpi} and applying the identity $\MCO{X}{c}{t_1}\MCO{X}{c}{t_2}=\MCO{X}{t_1}{t_2}\MCO{X}{c}{t_1}\MCO{X}{t_1}{t_2}$. The circuit requires 4 additional CNOT gates which already operate on neighbouring qubits.
\end{proof}

\begin{lemma}\label{lem:CU4_LNN_mid}
    Any $\MCO{V}{c}{\{t_1,t_2\}}$ gate with $V\in U(4)$ can be implemented in LNN architecture over 3 neighbouring qubits using 13 CNOT gates if $c$ is the central qubit.
\end{lemma}
\begin{proof}
The labels of $c,t_1$ in \qc{CU4_a2a} can be swapped by applying a SWAP gate on on each side of the circuit.
 
 The gates $\MCO{\rpisub{1}}{t_1}{c}\MCO{\rpisub{2}}{t_1}{t_2}\MCO{\text{SWAP}}{t_1}{c}$ can be applied using 3 CNOT gates using $\MCO{\rpisub{2}}{t_1}{t_2}\MCO{\text{SWAP}}{t_1}{c} = \MCO{\text{SWAP}}{t_1}{c}\MCO{\rpisub{2}}{c}{t_2}$, and $\MCO{\rpisub{1}}{t_1}{c}\MCO{\text{SWAP}}{t_1}{c} = \sqO{\rpisub{1}'}{c}\MCO{X}{c}{t_1}\MCO{X}{t_1}{c} \sqO{\rpisub{1}'}{t1}$, with $\rpisub{1}'$ obtained using the decomposition of $\MCO{\rpisub{1}}{t_1}{c}$ as \qc{crpi}. The gates $\MCO{\text{SWAP}}{t_1}{c}\MCO{\rpisub{7}}{t_1}{t_2}\MCO{\rpisub{6}}{t_1}{c}$ are applied similarly. 
 The remaining controlled gates are decomposed as in \lem{CU4_LNN_top}.
\end{proof}

While $\pi$-rotations can be useful when implemented directly, some physical qubit implementations do not support such rotations, and in that case $\rpi{\hat{v}}$ gates must be implemented using the available set of basic gates. However, as demonstrated in this paper, the $\pi$-rotation formalism can be useful as a convenient framework for circuit design. In this section we provide a few methods which can be used in many cases in order to convert a circuit which was built using $\pi$-rotations to an equivalent circuit which uses gates from the group $\{CNOT,R_{\hat{y}},R_{\hat{z}}\}$, and use those methods in order to decompose the presented controlled $U(4)$ gates.

\begin{lemma}\label{lem:Cpi_dec_yz}
    Any $\MCO{\rpi{\hat{v}}}{C}{t}$ gate can be implemented as $\sqO{\rz(\phi)}{t}\sqO{\rv{\psi}{\hat{y}}}{t}\MCO{X}{C}{t}\sqO{\rv{-\psi}{\hat{y}}}{t}\sqO{\rz(-\phi)}{t}$ for some angles $\phi,\psi$.
\end{lemma}
\begin{proof}
    According to \lem{MC_transform}, $\MCO{\rpi{\hat{v}}}{C}{t} = \sqO{\rpi{\hat{v}_M}}{t}\MCO{X}{C}{t}\sqO{\rpi{\hat{v}_M}}{t}$ with $\hat{v}_M\in M(\hat{v},\hat{x})$. Any operator can be decomposed using ZYX Euler (Tait–Bryan) 
    angles, and therefore, three angles $\phi,\psi,\zeta$ can be found such that $\rpi{\hat{v}_M} = 
    \rv{\phi}{\hat{z}}\rv{\psi}{\hat{y}}\rv{\zeta}{\hat{x}}$, and since $\rpi{\hat{v}_M}$ is Hermitian, $\rpi{\hat{v}_M} = \rv{-\zeta}{\hat{x}} 
    \rv{-\psi}{\hat{y}}
    \rv{-\phi}{\hat{z}}$ as well.
    The rotations around $\hat{x}$ can commute with the $\MCO{X}{C}{t}$ gate and cancel out, providing the required expression.
\end{proof}

The number of single qubit rotations can be reduced in the specific cases in which $\hat{v}$ is perpendicular to the $\hat{x}$, $\hat{y}$ or $\hat{z}$ axis.

\begin{lemma}\label{lem:Cpi_zz_yy}
If $\hat{v}=\hat{R}_{\hat{\sigma}}(\phi)\hat{\tau}$ for any unit vectors $\hat{\tau},\hat{\sigma}$ satisfying $\hat{\tau}\perp \hat{\sigma}$, then $\MCO{\rpi{\hat{v}}}{C}{t}=\sqO{\rv{\phi}{\hat{\sigma}}}{t}\MCO{\rpi{\hat{\tau}}}{C}{t}\sqO{\rv{-\phi}{\hat{\sigma}}}{t}$.
\end{lemma}
\begin{proof}
    Using \lem{MC_transform}, $\MCO{\rpi{\hat{v}}}{C}{t} = \sqO{\rpi{\hat{v}_M}}{t}\MCO{\rpi{\hat{\tau}}}{C}{t}\sqO{\rpi{\hat{v}_M}}{t}$ with $\hat{v}_M\in M(\hat{v},\hat{\tau})$. Since both $\hat{v},\hat{\tau}$ are perpendicular to $\hat{\sigma}$ and $\hat{v}=\hat{R}_{\hat{\sigma}}(\phi)\hat{\tau}$, we can choose $\hat{v}_M=\hat{R}_{\sigma}(\frac{\phi}{2})\hat{\tau}$. This vector is on the same plane as $\hat{v},\hat{\tau}$, and $\ang{\hat{v}}{\hat{v}_M}=\ang{\hat{\tau}}{\hat{v}_M}$ and therefore it satisfies $\hat{v}=\hat{R}_{\hat{v}_M}(\pi)\hat{\tau}$. Substituting  $\MCO{\rpi{\hat{\tau}}}{C}{t}=\sqO{\rpi{\hat{\tau}}}{t}\MCO{\rpi{\hat{\tau}}}{C}{t}\sqO{\rpi{\hat{\tau}}}{t}$ provides:
    \[
    \MCO{\rpi{\hat{v}}}{C}{t} = \sqO{\rpi{\hat{v}_M}}{t}\sqO{\rpi{\hat{\tau}}}{t}\MCO{\rpi{\hat{\tau}}}{C}{t}\sqO{\rpi{\hat{\tau}}}{t}\sqO{\rpi{\hat{v}_M}}{t}
    \]
    From \lem{2_rpi}, $\rpi{\hat{v}_M}\rpi{\hat{\tau}}=\rv{\phi}{\hat{\sigma}}$ and $\rpi{\hat{\tau}}\rpi{\hat{v}_M}=\rv{-\phi}{\hat{\sigma}}$.
\end{proof}
Using \lem{Cpi_zz_yy} with $\hat{\tau}=\hat{x}$ and $\hat{\sigma}\in \{\hat{y},\hat{z} \}$ allows one to implement any single-controlled $\pi$-rotation about a unit vector on the $xy$ or $xz$ plane using two $R_{\hat{z}}$ or two $R_{\hat{y}}$ gates, respectively, along with one CNOT gate.

We now wish to decompose \qc{CU4_a2a} using gates from the universal sets $\{CNOT,R_{\hat{y}},R_{\hat{z}}\}$, $\{CNOT,R_{\hat{v}}\}$ and $\{CNOT,\rpicl\}$. For each of these sets, we construct a circuit which can implement any controlled-$U(4)$ operator, and present the cost of the circuit in  all-to-all and in LNN architectures. 
\begin{theorem}\label{thm:csu4_rzry}
 Any $\MCO{V}{c}{\{t_1,t_2\}}, V\in U(4)$ gate can be implemented over three neighbouring qubits using $10$ $\ry$ gates and $15$ $\rz$ gates, in addition to $10$ or $13$ CNOT gates in all-to-all or LNN connectivity, respectively.
\end{theorem}
\begin{proof}
    According to \lem{CU4_a2a_lem}, any $\MCO{V}{c}{\{t_1,t_2\}}, V\in U(4)$ gate can be implemented using \qc{CU4_a2a} such that $\hat{v}_3,\hat{v}_5$ are on the $xz$ plane, and $\hat{v}_1,\hat{v}_2,\hat{v}_4$ are on the $xy$ plane.
    We can simply decompose \qc{CU4_a2a} in terms of $\{CNOT,\ry,\rz\}$ gates. \\
    First we note that $\hat{y}=\hat{R}_{\hat{z}}(\frac{\pi}{2})\hat{x}$, and $\hat{z}=\hat{R}_{\hat{y}}(-\frac{\pi}{2})\hat{x}$, thus according to \lem{Cpi_zz_yy}, $\MCO{iY}{t_1}{t_2}=\sqO{S}{t_2}\MCO{X}{t_1}{t_2}\sqO{S^\dagger}{t_2}\sqO{S}{t_1}$, and $\MCO{Z}{t_2}{t_1} = \sqO{\rv{-\frac{\pi}{2}}{\hat{y}}}{t_2}\MCO{X}{t_2}{t_1}\sqO{\rv{\frac{\pi}{2}}{\hat{y}}}{t_2}$. Any other controlled gate in the circuit can be implemented according to \lem{Cpi_zz_yy} if the axis is constrained to a specific plane, or according to \lem{Cpi_dec_yz} otherwise. Applying these to all gates and merging sequential rotations into one rotation provides the following structure:
\[
\scalebox{0.73}{
\Qcircuit @C=0.2em @R=0.1em @!R { 
	 	\nghost{{c} :  } & \lstick{{c} :  }  & \gate{\mathrm{\rz}}  & \ctrl{1}  & \qw  & \qw & \qw & \ctrl{1}  &  \qw & \qw & \qw & \ctrl{1} & \qw & \qw & \qw &  \ctrl{1} & \qw & \qw & \qw\\
	 	\nghost{{t}_{1} :  } & \lstick{{t}_{1} :  } & \gate{\mathrm{\rz}} & \targ\ar @{-} [1,0] &  \gate{\mathrm{R}} &  \ctrl{1}  & \gate{\mathrm{\ry}} & \targ\ar @{-} [1,0] & \gate{\mathrm{\ry}} & \targ & \gate{\mathrm{\ry}} & \targ & \gate{\mathrm{\ry}} & \ctrl{1} & \gate{\mathrm{R}} & \targ\ar @{-} [1,0] & \gate{\mathrm{\ry}} & \gate{\mathrm{\rz}} & \qw \\
	 	\nghost{{t}_{2} :  } & \lstick{{t}_{2} :  } & \gate{\mathrm{\rz}} & \targ  & \gate{\mathrm{R}} &  \targ & \gate{\mathrm{\rz}} & \targ & \gate{\mathrm{\rz}}  & \ctrl{-1} & \qw & \qw & \qw & \targ & \gate{\mathrm{R}} & \targ & \gate{\mathrm{\ry}} & \gate{\mathrm{\rz}} & \qw \\
 }
 }
 \qcref{CU4_a2a_zy}
\]
where we only indicate the axis of rotation if known, and do not specify the angles for simplicity. This circuit requires four arbitrary rotations which cost 
$4$ $\ry$ gates and $8$ $\rz$ gates using the $ZYZ$ decomposition. In addition this circuit uses $6$ $\ry$ gates, $7$ $\rz$ gates and $10$ CNOTs.
\\
 Following the steps in \lem{CU4_LNN_top} and \lem{CU4_LNN_mid}, it is clear that single qubit rotations are not added in the process of converting this circuit into its LNN versions which require 13 CNOT gates for any permutation over three neighbouring qubits.
\end{proof}

The decomposition in terms of $\{CNOT,R_{\hat{v}}\}$ gates follows simply from \qc{CU4_a2a_zy}, merging any consecutive rotations into one, and following the same logic of \thm{csu4_rzry} for the LNN structures. Those steps show the correctness of \thm{csu4_rv}.

\begin{theorem}\label{thm:csu4_rv}
    Any $\MCO{V}{c}{\{t_1,t_2\}}, V\in U(4)$ gate can be implemented over three neighbouring qubits using $15$ $R_{\hat{v}}$ gates, in addition to $10$ or $13$ CNOT gates in all-to-all or LNN connectivity respectively.
\end{theorem}
Alternatively, \qc{CU4_a2a} can be easily converted to a circuit which uses $\{CNOT,\rpicl\}$ gates by repeating the same steps as the ones used to produce \qc{CU4_a2a_zy}, but decomposing the last two controlled $\pi$-rotations according to \lem{MC_transform}. Then adding an $X$ gate on both sides of each target of the first five CNOT gates which are controlled by qubit $c$ allows to convert the $\rz /\ry$ gates to $\pi$-rotations with axis on the $xy/xz$ plane using the same logic as in the proof of \lem{Cpi_zz_yy}. The $R$ gates are then decomposed according to \lem{2_rpi}, with one axis chosen from the $xy$ plane, and the $\rz$ gate on the control is decomposed similarly as $X\rpisubu{}{z}$. The following circuit presents the decomposition in terms of the Hermitian universal set $\{CNOT,\rpicl\}$.
\[
 \resizebox{0.85\linewidth}{!}{
\Qcircuit @C=0.25em @R=0.15em @!R { 
	 	\nghost{{c} :  } & \lstick{{c} :  }  & \qw  & \ctrl{1} & \rpigatesubu{z}{} & \gate{\mathrm{X}}    & \qw & \qw & \ctrl{1}  &  \qw & \qw & \qw & \ctrl{1} & \qw & \qw & \qw & \qw &  \ctrl{1} & \qw & \qw \\
	 	\nghost{{t}_{1} :  } & \lstick{{t}_{1} :  } & \rpigatesubu{z}{} & \targ\ar @{-} [1,0] &  \rpigatesubu{z}{} & \rpigatesub{} &  \ctrl{1}  & \rpigatesubu{y}{} & \targ\ar @{-} [1,0] & \rpigatesubu{y}{} & \targ & \rpigatesubu{y}{} & \targ & \rpigatesubu{y}{} & \ctrl{1} & \rpigatesubu{z}{} & \rpigatesub{}& \targ\ar @{-} [1,0] & \rpigatesub{} & \qw \\
	 	\nghost{{t}_{2} :  } & \lstick{{t}_{2} :  } & \rpigatesubu{z}{} & \targ  & \rpigatesubu{z}{} & \rpigatesub{} &  \targ & \rpigatesubu{z}{} & \targ & \rpigatesubu{z}{}  & \ctrl{-1} & \qw & \qw & \qw & \targ & \rpigatesubu{z}{} & \rpigatesub{} & \targ & \rpigatesub{} & \qw \\
 }
 }
 \qcref{CU4_a2a_pi}
\]
\\
This circuit uses $20$ $\rpicl$ gates in total, while one is $X$ and $13$ are restricted to the $xy$ or the $xz$ plane. In the above circuit, it is implicit that according to the specific $V$, the axes corresponding to each of the $\Pi$ gates will be determined. This circuit can be converted into its LNN versions in a similar manner to the proof of \thm{csu4_rzry}. Those steps show the correctness of \thm{csu4_rpi}.

\begin{theorem}\label{thm:csu4_rpi}
    Any $\MCO{V}{c}{\{t_1,t_2\}}, V\in U(4)$ gate can be implemented over three neighboring qubits using $20$ $\rpicl$ gates, in addition to $10$ or $13$ CNOT gates in all-to-all or LNN connectivity, respectively.
\end{theorem}
\section{Proofs of $\Pi$ gate Lemmas}\label{apx:pi_lemmas}

\noindent Here we provide proofs for Lemmas mentioned in the main text.

\setcounter{lemma}{0}
\begin{lemma}
\label{lem:2_rpi_apx}
Any $\rv{\lambda}{\hat{v}}\in SU(2)$ operator can be implemented as $\rpi{\hat{v}_2}\rpi{\hat{v}_1}$ 
with $\hat{v}_1$ as any unit vector perpendicular to $\hat{v}$, and $\hat{v}_2 = \hat{R}_{\hat{v}}(\frac{\lambda}{2})\hat{v}_1$ with $\frac{\lambda}{2}\in (-\pi,\pi]$.
\end{lemma}
\begin{proof}
First, we focus on the case $\frac{\lambda}{2}\in (0,\pi)$ in which the ordered set $(\hat{v}_1,\hat{v}_2, \hat{v})$ is positively-oriented.
Since $\hat{v}_1\perp\hat{v}$, any rotation of $\hat{v}_1$ around $\hat{v}$ is perpendicular to $\hat{v}$ as well, and so $\hat{v}_2\perp\hat{v}$, and $\ang{\hat{v}_1}{\hat{v}_2} = \frac{\lambda}{2}$. Therefore, using vector arithmetic,
$
\hat{v}_1\cdot \hat{v}_2=\cos{\frac{\lambda}{2}}$, and $ \hat{v}_1\times\hat{v}_2=\hat{v}\sin{\frac{\lambda}{2}}
$. 
The following is achieved by matrix multiplication, followed by substitutions provided by the above equations, e.g. $[v_{1_x}v_{2_y}-v_{1_y}v_{2_x}]=[v_z\sin{\tfrac{\lambda}{2}}]$.
\[
\rpi{\hat{v}_2}\rpi{\hat{v}_1} =
\left(
\begin{matrix}
\cos{\tfrac{\lambda}{2}}-iv_z\sin{\tfrac{\lambda}{2}} & 
(-v_y-iv_x)\sin{\tfrac{\lambda}{2}}\\
(v_y-iv_x)\sin{\tfrac{\lambda}{2}} & 
\cos{\tfrac{\lambda}{2}}+iv_z\sin{\tfrac{\lambda}{2}}
\end{matrix}
\right)
\]
which is equivalent to $\rv{\lambda}{\hat{v}}$ as shown in \cite{nielsen_quantum_2002}.
For $\frac{\lambda}{2}\in (-\pi,0)$, the ordered set $(\hat{v}_1,\hat{v}_2, -\hat{v})$ is positively-oriented, and $\ang{\hat{v}_1}{\hat{v}_2} = \abs{\frac{\lambda}{2}} =-\frac{\lambda}{2}$. Thus, using the above, $\rpi{\hat{v}_2}\rpi{\hat{v}_1} = \rv{-\lambda}{-\hat{v}} = \rv{\lambda}{\hat{v}}$.
For the specific cases of $\frac{\lambda}{2}\in \{0,\pi\}$, we get $\rv{0}{\hat{v}}=\rpi{\hat{v}_1}\rpi{\hat{v}_1}=I$, and $\rv{2\pi}{\hat{v}}=\rpi{-\hat{v}_1}\rpi{\hat{v}_1}=-I$ as expected. Therefore, the statement holds for $\frac{\lambda}{2}\in (-\pi,\pi]$.
\end{proof}

\setcounter{lemma}{2}
\begin{lemma}\label{lem:1q_transform_apx} 
$
\rv{\lambda}{\hat{v}}
=
\rpi{\hat{v}_{M}}
\rv{\lambda}{\hat{v}'}
\rpi{\hat{v}_{M}}
$
for any angle $\lambda$, unit vectors $\hat{v},\hat{v}'$, and $\hat{v}_M \in M(\hat{v},\hat{v}')$.
\end{lemma}
\begin{proof}
According to \lem{2_rpi}, 
$\rv{\lambda}{\hat{v}}
=
\rpi{\hat{v}_2}
\rpi{\hat{v}_1}
$, such that $\hat{v}_1\perp\hat{v}$, and 
$\hat{v}_2=\hat{R}_{\hat{v}}(\frac{\lambda}{2})\hat{v}_1$.
Using \lem{sq_trans}, 
$
\rpi{\hat{v}_j}
=
\rpi{\hat{v}_{M}}
\rpi{\hat{v}'_j}
\rpi{\hat{v}_{M}}
$ with $j\in\{1,2\}$, and $\hat{v}'_j:=\hat{R}_{\hat{v}_{M}}(\pi)\hat{v}_j$. Substituting and canceling out two $\rpi{\hat{v}_{M}}$ operators provides
$\rv{\lambda}{\hat{v}}
=
\rpi{\hat{v}_{M}}
\rpi{\hat{v}'_2}
\rpi{\hat{v}'_1}
\rpi{\hat{v}_{M}}$.

Since the angles between vectors are conserved under rotations, and the same rotation matrix $\hat{R}_{\hat{v}_{M}}(\pi)$ is applied on $\hat{v}_1,\hat{v}_2,\hat{v}$ to produce $\hat{v}'_1,\hat{v}'_2,\hat{v}'$ respectively, we get that $\hat{v}'_1\perp\hat{v}'$, and 
$\hat{v}'_2=\hat{R}_{\hat{v}'}(\frac{\lambda}{2})\hat{v}'_1$. Therefore, $\rv{\lambda}{\hat{v}'}
=
\rpi{\hat{v}'_2}
\rpi{\hat{v}'_1}
$.
\end{proof}
\setcounter{lemma}{11}
\begin{lemma}\label{lem:rpi_if_1cx} 
    Any $\MCO{U}{c}{t}, U\in U(2)$ gate which can be implemented using one CNOT must satisfy $U=e^{i\psi}\rpi{\hat{v}}$.
\end{lemma}

\begin{proof}
Any circuit over two qubits with one $CNOT$ gate can be expressed in the following generic form:
    \[
    \MCO{U}{c}{t} =
    \sqO{U_4}{t}\sqO{U_3}{c}
    \MCO{X}{c}{t}
    \sqO{U_1}{c}\sqO{U_2}{t}
    \]
    The gates $U_{j\in\{1,2,3,4\}}\in U(2)$ can be written as
$
        \left(
    \begin{matrix}
        \alpha_j & \gamma_j\\
        \beta_j & \delta_j 
    \end{matrix}
    \right)$.
    We adjust the equation and apply each side on $\ket{00}$.
    \[
    \sqO{U^\dagger_4}{t}
    \MCO{U}{c}{t}
    \sqO{U^\dagger_2}{t}\ket{00}
    =
    \sqO{U_3}{c}
    \MCO{X}{c}{t}
    \sqO{U_1}{c}\ket{00}   
    \]
    Since $\MCO{U}{c}{t}$ applies the identity if the computational basis state of the control is $\ket{0}$, this is equivalent to
    \[
    \ket{0}_c(U^\dagger_4U^\dagger_2\ket{0})_t
    = 
        \sqO{U_3}{c}
    (\alpha_1\ket{00}+\beta_1\ket{11})
    =
    \]
    \[
    \ket{0}_c
    (\alpha_1\alpha_3\ket{0} +
    \beta_1\gamma_3\ket{1})_t
    +
    \ket{1}_c
    (\alpha_1\beta_3\ket{0} +
    \beta_1\delta_3\ket{1})_t       
    \]
    
    This provides the constraints 
    $U^\dagger_4U^\dagger_2 = \alpha_1\alpha_3I +
    \beta_1\gamma_3X$ and
     $\lvert\alpha_1\rvert\lvert\beta_3\rvert = \lvert\beta_1\rvert\lvert\delta_3\rvert = 0$.
    The identities $\lvert\alpha_j\rvert = \lvert\delta_j\rvert$,  $\lvert\beta_j\rvert = \lvert\gamma_j\rvert$ and $\lvert\alpha_j\rvert^2 + \lvert\beta_j\rvert^2 = 1$ are used to show that there are only two possible solutions up to a global phase:
    \[
    (U_4,U_3,U_1) =
\begin{cases}
    (U_2^\dagger, P(\psi_3), P(\psi_1))     & \alpha_1 \not=0\\
    (U_2^\dagger X, P(\psi'_3)X, P(\psi'_1)X) & \alpha_1 = 0
\end{cases}
    \]

    And therefore any $\MCO{U}{c}{t}$ which can be implemented using a single CNOT can be written as
\[
\MCO{U}{c}{t} = \sqO{P(\psi)}{c}\sqO{U_2^\dagger}{t}\MCO{X}{c}{t}\sqO{U_2}{t}
\]    
    this is obtained by substituting each of the solutions and using the identities $\sqO{X}{c}\sqO{P(\psi)}{c} = \sqO{P(-\psi)}{c}\sqO{X}{c}$, $\MCO{X}{c}{t}\sqO{P(\psi)}{c} = \sqO{P(\psi)}{c}\MCO{X}{c}{t}$  and  $\sqO{X}{c}
    \MCO{X}{c}{t}
    \sqO{X}{c}\sqO{X}{t} = \MCO{X}{c}{t}$.

Therefore, it must hold that $U = e^{i\psi}U^{\dagger}_2 X U_2$. According to \lem{2_rpi}, we can find $\hat{v}_1,\hat{v}_2$ such that $U_2 = \rpi{\hat{v}_2}\rpi{\hat{v}_1}$, and so $U = e^{i\psi}\rpi{\hat{v}_1}\rpi{\hat{v}_2} X \rpi{\hat{v}_2}\rpi{\hat{v}_1}$. According to \lem{sq_trans}, one can find unit vectors $\hat{v}'_2,\hat{v}$ such that $\rpi{\hat{v}_2} X \rpi{\hat{v}_2} = \rpi{\hat{v}'_2} $ , and  $\rpi{\hat{v}_1} \rpi{\hat{v}'_2} \rpi{\hat{v}_1} = \rpi{\hat{v}}$. Therefore, $U = e^{i\psi}\rpi{\hat{v}}$.
\end{proof}

\section{Toffoli decompositions}\label{apx:toffolispi}
The following circuit demonstrates the decomposition of the Toffoli gate using only gates from $\{CNOT,H,X,\rpisub{T}\}$.
\[
\scalebox{0.8}{
\Qcircuit @C=1.0em @R=0.73em @!R { 
	 	\nghost{} & \lstick{} & \targ & \qw \\
	 	\nghost{} & \lstick{} & \ctrl{-1} & \qw \\
	 	\nghost{} & \lstick{} & \ctrl{-1} & \qw \\
 }
\hspace{5mm}\raisebox{-7mm}{=}\hspace{0mm}
\Qcircuit @C=0.5em @R=0.2em @!R { 
	 	\nghost{} & \lstick{} & \gate{\mathrm{H}} & \ctrl{1} & \qw & \ctrl{1} & \ctrl{2} & \gate{\mathrm{X}} & \qw & \rpigatesub{T}  & \ctrl{2} & \gate{\mathrm{H}} & \qw \\
	 	\nghost{} & \lstick{} & \ctrl{1} & \targ & \rpigatesub{T} & \targ & \qw & \rpigatesub{T} & \ctrl{1} & \qw & \qw & \qw & \qw \\
	 	\nghost{} & \lstick{} & \targ & \qw & \rpigatesub{T}  & \qw & \targ & \rpigatesub{T} & \targ & \rpigatesub{T} & \targ & \rpigatesub{T} & \qw \\
 }
 }
\]
This decomposition can be easily achieved by converting the well known implementation of the Toffoli gate, replacing $T$ and $T^{\dagger}$ gates with $\rpisub{T}X$ and  $X\rpisub{T}$ respectively, and canceling out $X$ gates.\\
The following circuit implements the Toffoli gate up to a global phase using the minimal universal Hermitian set $\{CNOT,H,\rpisub{T}\}$ with seven CNOT, seven $\rpisub{T}$ and two $H$ gates.
\[
\resizebox{0.99\linewidth}{!}{
\Qcircuit @C=1.0em @R=0.73em @!R { 
	 	\nghost{} & \lstick{} & \ctrl{1} & \qw \\
	 	\nghost{} & \lstick{} & \ctrl{1} & \qw \\
	 	\nghost{} & \lstick{} & \targ & \qw \\
 }
\hspace{5mm}\raisebox{-7mm}{=}\hspace{0mm}
\Qcircuit @C=0.5em @R=0.2em @!R { 
	 	\nghost{} & \lstick{} & \qw & \qw & \ctrl{1} & \qw & \ctrl{1} & \targ & \rpigatesub{T} & \targ & \qw & \rpigatesub{T}  & \targ & \rpigatesub{T} & \qw \\
	 	\nghost{} & \lstick{} & \qw & \ctrl{1} & \targ & \rpigatesub{T} & \targ & \qw & \rpigatesub{T} & \ctrl{-1} & \ctrl{1} & \qw & \qw & \qw & \qw \\
	 	\nghost{} & \lstick{} & \gate{\mathrm{H}} & \targ & \qw & \rpigatesub{T}  & \qw & \ctrl{-2} & \qw & \qw & \targ & \rpigatesub{T} & \ctrl{-2} & \gate{\mathrm{H}} & \qw \\
 }
 }
\]

\begin{thebibliography}{10}

\bibitem{nielsen_quantum_2002}
Michael~A. Nielsen and Isaac Chuang.
\newblock ``Quantum {Computation} and {Quantum} {Information}''.
\newblock \href{https://dx.doi.org/10.1119/1.1463744}{American Journal of Physics {\bf 70}, 558--559}~(2002).

\bibitem{bar-gill_solid-state_2013}
N.~Bar-Gill, L.~M. Pham, A.~Jarmola, D.~Budker, and R.~L. Walsworth.
\newblock ``Solid-state electronic spin coherence time approaching one second''.
\newblock \href{https://dx.doi.org/10.1038/ncomms2771}{Nature Communications {\bf 4}, 1743}~(2013).

\bibitem{haroche_collge_2012}
Serge Haroche.
\newblock ``Collge de {France} abroad {Lectures} {Quantum} information with real or artificial atoms and photons in cavities''~(2012).
\newblock
Centre for Quantum Technologies, National University of Singapore.
\newblock \href{https://www.cqt.sg/highlight/2012-03-college-de-france-lectures-serge-haroche/}{https://www.cqt.sg/highlight/2012-03-college-de-france-lectures-serge-haroche/}.

\bibitem{husain_further_2013}
Sami Husain, Minaru Kawamura, and Jonathan~A. Jones.
\newblock ``Further analysis of some symmetric and antisymmetric composite pulses for tackling pulse strength errors''.
\newblock \href{https://dx.doi.org/10.1016/j.jmr.2013.02.007}{Journal of Magnetic Resonance {\bf 230}, 145--154}~(2013).

\bibitem{low_optimal_2014}
Guang~Hao Low, Theodore~J. Yoder, and Isaac~L. Chuang.
\newblock ``Optimal arbitrarily accurate composite pulse sequences''.
\newblock \href{https://dx.doi.org/10.1103/PhysRevA.89.022341}{Physical Review A {\bf 89}, 022341}~(2014).

\bibitem{tycko_fixed_1985}
R.~Tycko, A.~Pines, and J.~Guckenheimer.
\newblock ``Fixed point theory of iterative excitation schemes in {NMR}''.
\newblock \href{https://dx.doi.org/10.1063/1.449228}{The Journal of Chemical Physics {\bf 83}, 2775--2802}~(1985).

\bibitem{jones_nested_2013}
Jonathan~A. Jones.
\newblock ``Nested composite {NOT} gates for quantum computation''.
\newblock \href{https://dx.doi.org/10.1016/j.physleta.2013.08.040}{Physics Letters A {\bf 377}, 2860--2862}~(2013).

\bibitem{viola_dynamical_1999}
Lorenza Viola, Emanuel Knill, and Seth Lloyd.
\newblock ``Dynamical {Decoupling} of {Open} {Quantum} {Systems}''.
\newblock \href{https://dx.doi.org/10.1103/PhysRevLett.82.2417}{Physical Review Letters {\bf 82}, 2417--2421}~(1999).

\bibitem{uhrig_keeping_2007}
Götz~S. Uhrig.
\newblock ``Keeping a {Quantum} {Bit} {Alive} by {Optimized} pi -{Pulse} {Sequences}''.
\newblock \href{https://dx.doi.org/10.1103/PhysRevLett.98.100504}{Physical Review Letters {\bf 98}, 100504}~(2007).

\bibitem{souza_robust_2011}
Alexandre~M. Souza, Gonzalo~A. Álvarez, and Dieter Suter.
\newblock ``Robust {Dynamical} {Decoupling} for {Quantum} {Computing} and {Quantum} {Memory}''.
\newblock \href{https://dx.doi.org/10.1103/PhysRevLett.106.240501}{Physical Review Letters {\bf 106}, 240501}~(2011).

\bibitem{souza_robust_2012}
Alexandre~M. Souza, Gonzalo~A. Álvarez, and Dieter Suter.
\newblock ``Robust dynamical decoupling''.
\newblock \href{https://dx.doi.org/10.1098/rsta.2011.0355}{Philosophical Transactions of the Royal Society A: Mathematical, Physical and Engineering Sciences {\bf 370}, 4748--4769}~(2012).

\bibitem{nemkov_efficient_2023}
Nikita~A. Nemkov, Evgeniy~O. Kiktenko, Ilia~A. Luchnikov, and Aleksey~K. Fedorov.
\newblock ``Efficient variational synthesis of quantum circuits with coherent multi-start optimization''.
\newblock \href{https://dx.doi.org/10.22331/q-2023-05-04-993}{Quantum {\bf 7}, 993}~(2023).

\bibitem{householder_unitary_1958}
Alston~S. Householder.
\newblock ``Unitary {Triangularization} of a {Nonsymmetric} {Matrix}''.
\newblock \href{https://dx.doi.org/10.1145/320941.320947}{Journal of the ACM {\bf 5}, 339--342}~(1958).

\bibitem{barenco_elementary_1995}
Adriano Barenco, Charles~H. Bennett, Richard Cleve, David~P. DiVincenzo, Norman Margolus, Peter Shor, Tycho Sleator, John~A. Smolin, and Harald Weinfurter.
\newblock ``Elementary gates for quantum computation''.
\newblock \href{https://dx.doi.org/10.1103/PhysRevA.52.3457}{Physical Review A {\bf 52}, 3457--3467}~(1995).

\bibitem{brezov_vector_2012}
Danail Brezov, Clementina Mladenova, and Ivaïlo Mladenov.
\newblock ``Vector {Decompositions} of {Rotations}''.
\newblock \href{https://dx.doi.org/10.7546/jgsp-28-2012-67-103}{Journal of Geometry and Symmetry in Physics {\bf 28}, 67--103}~(2012).

\bibitem{donchev_compositions_2015}
V.~D. Donchev, C.~D. Mladenova, and I.~M. Mladenov.
\newblock ``On the compositions of rotations''.
\newblock \href{https://dx.doi.org/10.1063/1.4934315}{AIP Conference Proceedings {\bf 1684}, 080004}~(2015).

\bibitem{bravyi_universal_2005}
Sergey Bravyi and Alexei Kitaev.
\newblock ``Universal quantum computation with ideal {Clifford} gates and noisy ancillas''.
\newblock \href{https://dx.doi.org/10.1103/PhysRevA.71.022316}{Physical Review A {\bf 71}, 022316}~(2005).

\bibitem{zindorf_multi-controlled_2025}
Ben Zindorf and Sougato Bose.
\newblock ``Multi-{Controlled} {Quantum} {Gates} in {Linear} {Nearest} {Neighbor}''~(2025).
\newblock \href{https://dx.doi.org/10.48550/arXiv.2506.00695}{arXiv:2506.00695}.

\bibitem{wimperis_iterative_1991}
Stephen Wimperis.
\newblock ``Iterative schemes for phase-distortionless composite 180° pulses''.
\newblock \href{https://dx.doi.org/10.1016/0022-2364(91)90043-S}{Journal of Magnetic Resonance (1969) {\bf 93}, 199--206}~(1991).

\bibitem{wimperis_broadband_1994}
S.~Wimperis.
\newblock ``Broadband, {Narrowband}, and {Passband} {Composite} {Pulses} for {Use} in {Advanced} {NMR} {Experiments}''.
\newblock \href{https://dx.doi.org/10.1006/jmra.1994.1159}{Journal of Magnetic Resonance, Series A {\bf 109}, 221--231}~(1994).

\bibitem{gevorgyan_ultrahigh-fidelity_2021}
Hayk~L. Gevorgyan and Nikolay~V. Vitanov.
\newblock ``Ultrahigh-fidelity composite rotational quantum gates''.
\newblock \href{https://dx.doi.org/10.1103/PhysRevA.104.012609}{Physical Review A {\bf 104}, 012609}~(2021).

\bibitem{shim_single-qubit_2013}
Yun-Pil Shim, Jianjia Fei, Sangchul Oh, Xuedong Hu, and Mark Friesen.
\newblock ``Single-qubit gates in two steps with rotation axes in a single plane''~(2013).
\newblock 
\href{https://dx.doi.org/10.48550/arXiv.1303.0297}{arXiv:1303.0297}.

\bibitem{vale_circuit_2024}
Rafaella Vale, Thiago Melo~D. Azevedo, Ismael C.~S. Araújo, Israel~F. Araujo, and Adenilton~J. da~Silva.
\newblock ``Circuit {Decomposition} of {Multicontrolled} {Special} {Unitary} {Single}-{Qubit} {Gates}''.
\newblock \href{https://dx.doi.org/10.1109/TCAD.2023.3327102}{IEEE Transactions on Computer-Aided Design of Integrated Circuits and Systems {\bf 43}, 802--811}~(2024).

\bibitem{rosa_optimizing_2025}
Evandro C.~R. Rosa, Eduardo~I. Duzzioni, and Rafael~de Santiago.
\newblock ``Optimizing {Gate} {Decomposition} for {High}-{Level} {Quantum} {Programming}''.
\newblock \href{https://dx.doi.org/10.22331/q-2025-03-12-1659}{Quantum {\bf 9}, 1659}~(2025).

\bibitem{zindorf_efficient_2025}
Ben Zindorf and Sougato Bose.
\newblock ``Efficient implementation of multicontrolled quantum gates''.
\newblock \href{https://dx.doi.org/10.1103/8blx-nfcr}{Physical Review Applied {\bf 24}, 044030}~(2025).

\bibitem{gwinner_benchmarking_2021}
Jan Gwinner, Marcin Briański, Wojciech Burkot, Lukasz Czerwiński, and Vladyslav Hlembotskyi.
\newblock ``Benchmarking 16-element quantum search algorithms on superconducting quantum processors''~(2021).
\newblock 
\href{https://dx.doi.org/10.48550/arXiv.2007.06539}{arXiv:2007.06539}.

\bibitem{nakanishi_quantum-gate_2021}
Ken~M. Nakanishi, Takahiko Satoh, and Synge Todo.
\newblock ``Quantum-gate decomposer''~(2021).
\newblock
\href{https://dx.doi.org/10.48550/arXiv.2109.13223}{arXiv:2109.13223}.

\bibitem{nakanishi_decompositions_2024}
Ken~M. Nakanishi, Takahiko Satoh, and Synge Todo.
\newblock ``Decompositions of multiple controlled- {Z} gates on various qubit-coupling graphs''.
\newblock \href{https://dx.doi.org/10.1103/PhysRevA.110.012604}{Physical Review A {\bf 110}, 012604}~(2024).

\bibitem{m_q_cruz_shallow_2024}
Pedro M.~Q.~Cruz and Bruno Murta.
\newblock ``Shallow unitary decompositions of quantum {Fredkin} and {Toffoli} gates for connectivity-aware equivalent circuit averaging''.
\newblock \href{https://dx.doi.org/10.1063/5.0187026}{APL Quantum {\bf 1}, 016105}~(2024).

\bibitem{maslov_depth_2022}
Dmitri Maslov and Ben Zindorf.
\newblock ``Depth {Optimization} of {CZ}, {CNOT}, and {Clifford} {Circuits}''.
\newblock \href{https://dx.doi.org/10.1109/TQE.2022.3180900}{IEEE Transactions on Quantum Engineering {\bf 3}, 1--8}~(2022).

\bibitem{allende_synthesis_2024}
Carolina Allende, André~Fonseca de~Olivera, and Efrain Buksman.
\newblock ``Synthesis of quantum circuits based on supervised learning and correlations''.
\newblock \href{https://dx.doi.org/10.1007/s11128-024-04426-6}{Quantum Information Processing {\bf 23}, 204}~(2024).

\bibitem{shende_synthesis_2006}
V.V. Shende, S.S. Bullock, and I.L. Markov.
\newblock ``Synthesis of quantum-logic circuits''.
\newblock \href{https://dx.doi.org/10.1109/TCAD.2005.855930}{IEEE Transactions on Computer-Aided Design of Integrated Circuits and Systems {\bf 25}, 1000--1010}~(2006).

\bibitem{shende_minimal_2004}
Vivek~V. Shende, Igor~L. Markov, and Stephen~S. Bullock.
\newblock ``Minimal universal two-qubit controlled-{NOT}-based circuits''.
\newblock \href{https://dx.doi.org/10.1103/PhysRevA.69.062321}{Physical Review A {\bf 69}, 062321}~(2004).

\bibitem{vatan_optimal_2004}
Farrokh Vatan and Colin Williams.
\newblock ``Optimal {Quantum} {Circuits} for {General} {Two}-{Qubit} {Gates}''.
\newblock \href{https://dx.doi.org/10.1103/PhysRevA.69.032315}{Physical Review A {\bf 69}, 032315}~(2004).

\bibitem{vidal_universal_2004}
G.~Vidal and C.~M. Dawson.
\newblock ``Universal quantum circuit for two-qubit transformations with three controlled-{NOT} gates''.
\newblock \href{https://dx.doi.org/10.1103/PhysRevA.69.010301}{Physical Review A {\bf 69}, 010301}~(2004).

\end{thebibliography}
\end{document}